\documentclass[12pt,canadian]{article}
\usepackage[latin9]{inputenc}
\usepackage[letterpaper]{geometry}
\PassOptionsToPackage{ruled, vlined, linesnumbered}{algorithm2e}
\usepackage{color}
\usepackage{babel}
\usepackage{array}
\usepackage{rotating}
\usepackage{rotfloat}
\usepackage{booktabs}
\usepackage{algorithm2e}
\usepackage{amsmath}
\usepackage{amsthm}
\usepackage{amssymb}
\usepackage{graphicx}
\usepackage{setspace}
\usepackage[authoryear]{natbib}
\usepackage{xargs}[2008/03/08]
\usepackage[unicode=true,pdfusetitle,
 bookmarks=true,bookmarksnumbered=false,bookmarksopen=false,
 breaklinks=false,pdfborder={0 0 0},pdfborderstyle={},backref=false,colorlinks=true]
 {hyperref}
\hypersetup{
 linkcolor=pigment, citecolor=pigment, urlcolor=pigment}

\makeatletter

\providecommand{\tabularnewline}{\\}
\newcommand{\lyxdot}{.}

\newenvironment{lyxlist}[1]
	{\begin{list}{}
		{\settowidth{\labelwidth}{#1}
		 \setlength{\leftmargin}{\labelwidth}
		 \addtolength{\leftmargin}{\labelsep}
		 }}
	{\end{list}}
\theoremstyle{definition}
\newtheorem{defn}{\protect\definitionname}
\theoremstyle{plain}
\newtheorem{lem}{\protect\lemmaname}
\theoremstyle{plain}
\newtheorem{cor}{\protect\corollaryname}
\theoremstyle{remark}
\newtheorem{rem}{\protect\remarkname}

\definecolor{pigment}{rgb}{0.2, 0.2, 0.6}
\usepackage{arydshln}
\usepackage{breakcites}
\usepackage{breakurl}
\usepackage{amsfonts}

\renewcommand\paragraph{%
        \@startsection{paragraph}{4}{0mm}%
           {-\baselineskip}%
           {.5\baselineskip}%
           {\normalfont\normalsize\bfseries}}

\usepackage[font=footnotesize]{caption}

\theoremstyle{definition}
\newtheorem*{objective}{Objective}

\@ifundefined{showcaptionsetup}{}{%
 \PassOptionsToPackage{caption=false}{subfig}}
\usepackage{subfig}
\makeatother

\providecommand{\corollaryname}{Corollary}
\providecommand{\definitionname}{Definition}
\providecommand{\lemmaname}{Lemma}
\providecommand{\remarkname}{Remark}

\begin{document}
\global\long\def\var#1{\mathbb{V}\left(#1\right)}%
\global\long\def\pr#1{\mathbb{P}\left(#1\right)}%
\global\long\def\ev#1{\mathbb{E}\left(#1\right)}%

\global\long\def\smft#1#2{\stackrel[#1]{#2}{\sum}}%
\global\long\def\smo#1{\underset{#1}{\sum}}%
\global\long\def\prft#1#2{\stackrel[#1]{#2}{\prod}}%
\global\long\def\pro#1{\underset{#1}{\prod}}%
\global\long\def\uno#1{\underset{#1}{\bigcup}}%

\global\long\def\order#1{\mathcal{O}\left(#1\right)}%
\global\long\def\R{\mathbb{R}}%
\global\long\def\Q{\mathbb{Q}}%
\global\long\def\N{\mathbb{N}}%
\global\long\def\H{\mathbb{H}}%
\global\long\def\F{\mathcal{F}}%

\global\long\def\mathtext#1{\mathrm{#1}}%

\global\long\def\maxo#1{\underset{#1}{\max\,}}%
\global\long\def\argmaxo#1{\underset{#1}{\mathtext{argmax}\,}}%
\global\long\def\mino#1{\underset{#1}{\min\,}}%
\global\long\def\argmino#1{\underset{#1}{\mathtext{argmin}\,}}%
\global\long\def\limo#1#2{\underset{#1\rightarrow#2}{\lim}}%
\global\long\def\supo#1{\underset{#1}{\sup}}%
\global\long\def\info#1{\underset{#1}{\inf}}%

\global\long\def\b#1{\boldsymbol{#1}}%
\global\long\def\ol#1{\overline{#1}}%
\global\long\def\ul#1{\underline{#1}}%

\global\long\def\argparentheses#1{\mathopen{\left(#1\right)}\mathclose{}}%

\newcommandx\der[3][usedefault, addprefix=\global, 1=, 2=, 3=]{\frac{d^{#2}#3}{d#1^{#2}}}%
\newcommandx\pder[3][usedefault, addprefix=\global, 1=, 2=]{\frac{\partial^{#2}#3}{\partial#1^{#2}}}%
\global\long\def\intft#1#2#3#4{\int\limits _{#1}^{#2}#3d#4}%
\global\long\def\into#1#2#3{\underset{#1}{\int}#2d#3}%

\global\long\def\th{\theta}%
\global\long\def\la#1{~#1~}%
\global\long\def\laq{\la =}%
\global\long\def\normal#1#2{\mathcal{N}\left(#1,\,#2\right)}%
\global\long\def\uniform#1#2{\mathcal{U}\left(#1,\,#2\right)}%
\global\long\def\I#1#2{\mbox{I}_{#1}\left(#2\right)}%
\global\long\def\chisq#1{\chi_{#1}^{2}}%
\global\long\def\dar{\,\Longrightarrow\,}%
\global\long\def\dal{\,\Longleftarrow\,}%
\global\long\def\dad{\,\Longleftrightarrow\,}%
\global\long\def\norm#1{\left\Vert #1\right\Vert }%
\global\long\def\code#1{\mathtt{#1}}%
\global\long\def\descr#1#2{\underset{#2}{\underbrace{#1}}}%
\global\long\def\NB{\mathcal{NB}}%
\global\long\def\BNB{\mathcal{BNB}}%
\global\long\def\e#1{\text{e}_{#1}}%
\global\long\def\P{\mathbb{P}}%
\global\long\def\pb#1{\Bigg(#1\Bigg)}%

\newcommandx\T[1][usedefault, addprefix=\global, 1=]{\mathtext T_{#1}}%
\global\long\def\origs{O}%
\global\long\def\dests{D}%
\newcommandx\partner[3][usedefault, addprefix=\global, 1=]{\mathrm{partner}_{#1}\argparentheses{#2;\,#3}}%
\newcommandx\parent[2][usedefault, addprefix=\global, 1=]{\mathtext{parent}_{#1}\argparentheses{#2}}%
\newcommandx\predecessor[2][usedefault, addprefix=\global, 1=]{\mathtext{predecessor}_{#1}\argparentheses{#2}}%
\newcommandx\successor[2][usedefault, addprefix=\global, 1=]{\mathtext{successor}_{#1}\argparentheses{#2}}%
\global\long\def\reach#1{\mathtext{reach}\argparentheses{#1}}%
\global\long\def\cost#1{\mathtext{cost}\argparentheses{#1}}%
\newcommandx\dd[3][usedefault, addprefix=\global, 1=]{d_{#1}\argparentheses{#2,#3}}%
\global\long\def\l#1{l\argparentheses{#1}\mathclose{}}%

\global\long\def\evia{E_{\text{via}}}%
\global\long\def\vvia{V_{\text{via}}}%

\title{Locally Optimal Routes for Route Choice Sets}
\author{Samuel M. Fischer\thanks{Department of Mathematical and Statistical Sciences, University of
Alberta, Edmonton, AB T6G 2G1, Canada. Email: samuel.fischer@ualberta.ca}}
\maketitle
\begin{abstract}
Route choice is often modelled as a two-step procedure in which travellers
choose their routes from small sets of promising candidates. Many
methods developed to identify such choice sets rely on assumptions
about the mechanisms behind the route choice and require corresponding
data sets. Furthermore, existing approaches often involve considerable
complexity or perform many repeated shortest path queries. This makes
it difficult to apply these methods in comprehensive models with numerous
origin-destination pairs. In this paper, we address these issues by
developing an algorithm that efficiently identifies locally optimal
routes. Such paths arise from travellers acting rationally on local
scales, whereas unknown factors may affect the routes on larger scales.
Though methods identifying locally optimal routes are available already,
these algorithms rely on approximations and return only few, heuristically
chosen paths for specific origin-destination pairs. This conflicts
with the demands of route choice models, where an exhaustive search
for many origins and destinations would be necessary. We therefore
extend existing algorithms to return (almost) all admissible paths
between a large number of origin-destination pairs. We test our algorithm
and its applicability in route choice models on the road network of
the Canadian province British Columbia and empirical data collected
in this province.
\end{abstract}
\begin{description}
\item [{Keywords:}] alternative paths; choice set; computer algorithm;
local optimality; road network; route choice.
\end{description}

\section{Introduction}

Route choice models have important applications in transportation
network planning \citep{yang_models_1998}, traffic control \citep{mahmassani_dynamic_2001},
and even epidemiology and ecology \citep{fischer_hybrid_2020}. Route
choice models can be classified as either perfect rationality models
or bounded rationality models \citep{di_boundedly_2016}.  In perfect rationality models \citep{sheffi_urban_1984},
travellers are assumed to have complete information and choose their
routes optimally according to some goodness criterion, whereas bounded
rationality models \citep{simon_models_1957} take information constraints
and the complexity of the optimization process into account.  Though
both perfect rationality models and bounded rationality models have
been used in route choice modelling, bounded rationality models have
been found to fit observed data better \citep{di_boundedly_2016}.

Many bounded rationality models consider route choice as a two-stage
process: first, a so-called ``choice set'' of potentially good routes
is generated, and second, a route from the choice set is chosen according
to some goodness measure \citep{ben-akiva_modeling_1984}. This
approach is motivated through travellers' limited ability to consider
all possible paths. Instead, they may heuristically identify a small
set of routes from which they choose the seemingly best. Besides this
conceptual reasoning, the two-step model has computational advantages,
as the choice sets can be generated based on simple heuristics, while
complex models may be applied to determine travellers' preferences
for the identified routes. Therefore, the two-stage process is widely
used in route choice modelling \citep{prato_route_2009}.

Most of the approaches to identify route choice sets are based on
a combination of the optimality assumption, the constraint assumption,
and the stochasticity assumption.
\begin{itemize}
\item According to the optimality assumption, travellers choose routes optimally
according to some criterion, which could be based on route characteristics
(e.g. travel costs and travel time), or on scenarios (e.g. that the
travel time on the shortest route increases). Examples include the
link labelling approach \citep{ben-akiva_modeling_1984}, link elimination
\citep{azevedo_algorithm_1993}, and link penalty \citep{de_la_barra_multi-dimensional_1993}.
\item According to the constraint assumption, travellers consider all paths
whose quality exceeds a certain minimal value (e.g. acyclic paths
not more than $25\%$ longer than the shortest route). This assumption
motivates constrained enumeration methods \citep{prato_applying_2006}.
\item The stochasticity assumption accounts for the possibility of stochastic
fluctuations of route characteristics (e.g. through traffic jams or
accidents) or error-prone information. Often, stochastic route choice
sets are computed based on the optimality principle applied to a randomly
perturbed graph \citep[see][]{bovy_modelling_2009}.
\end{itemize}

Though each of the assumptions mentioned above has a sound mechanistic
justification, they require that the heuristic that travellers use
to identify potentially suitable paths is known and that corresponding
data are available. However, if travellers choose a route for unknown
reasons, e.g. because they desire to drive via some intermediate destination,
their routes would be difficult to consider with the common methods.
The natural solution would be to increase the set of generated routes
by relaxing constraints or modelling more mechanisms explicitly. However,
in comprehensive and large-scale route choice models, many origin-destination
pairs may have to be considered, making it costly or even infeasible
to work with large choice sets. Thus, it would be desirable to characterize
choice sets based on a more general but sufficiently restrictive criterion
that does not require knowledge or data of the specific mechanism
behind route choices.

A potentially suitable criterion is \emph{local} optimality. A route
is locally optimal if all its short (``local'') subsections are
optimal, respectively, according to a given measure. For example,
if travel time is the applied goodness criterion, a locally optimal
route would not contain local detours.

The rationale behind the principle of local optimality is that the
factors impacting travellers' routing decisions may differ dependent
on the spatial scale. Tourists, for example, may want to drive along
the shortest route locally but plan their trip globally to include
a number of sights. Other travellers may want to drive along the quickest
routes locally while minimizing the overall fuel consumption.  Yet
others may have a limited horizon of perfect information and act rationally
within this horizon only. Independent of the specific mechanism behind
travellers' route choices on the large scale, it is possible to characterize
many choice candidates as locally optimal routes. 

A potential problem with considering locally optimal routes is that
the set of locally optimal routes between an origin and a destination
can be very large and include zig-zag routes, which may seem unnatural.
A possible solution is to focus on so-called \emph{single-via paths}.
A single-via path is the shortest path via a given intermediate location.

Since not all locally optimal paths are single-via paths, restricting
the focus on single-via paths excludes some potentially suitable paths
from the choice set. However, single-via paths have a reasonable mechanistic
justification through travellers choosing intermediate destinations,
and the reduced choice sets are likely to include most of the routes
that travellers would reasonably choose. Since the reduced sets contain
relatively few elements, sophisticated models can be used for the
second decision stage, in which a route is chosen from the choice
set. Therefore, constraining the search for locally optimal routes
on single-via paths may lead to overall better fitting route choice
models.

To date, methods identifying locally optimal single-via paths have
been developed with the objective to suggest multiple routes to travellers
\citep{abraham_alternative_2013,delling_customizable_2015,luxen_candidate_2015,kliemann_route_2016}.
Such suggestions of alternative routes are a common feature in routing
software, such as Google Maps or Bing Maps. However, route choice
models have different demands than routing software, as travellers'
decisions shall be \emph{modelled} or \emph{predicted} rather than
\emph{facilitated}.

Route planning software seeks to compute a small number of high-quality
paths that travellers may want to choose. Here, computational speed
is more important than rigorous application of specific criteria characterizing
the returned paths. In contrast, route choice models should
consider \emph{all} routes that travellers may take, and rigorous
application of modelling assumptions is key to allow mechanistic inference
and to make models portable. In addition, route choice models may
consider \emph{multiple} origins and destinations. Therefore, many
algorithms designed to facilitate route planning cannot be directly
applied to identify route choice sets.

In this paper, we bridge this gap by extending an algorithm originally
designed for route planning. The algorithm REV by \citet{abraham_alternative_2013}
searches a small number of ``good'' locally optimal paths between
a single origin-destination pair. To this end, the algorithm uses an approximation
causing some locally optimal paths to be misclassified as suboptimal.

Our extended algorithm overcomes these limitations. Unlike REV, our
algorithm returns (almost) \emph{all }admissible paths between a \emph{set}
of origins and a \emph{set} of destinations. Therefore, we call our
algorithm REVC, the ``C'' emphasizing the attempted \textbf{c}omplete
search. REVC identifies locally optimal routes with arbitrarily high
precision. That is, the algorithm may falsely reject some locally
optimal single-via routes, but the error can be arbitrarily reduced
by cost of computational speed. As the execution time of REVC depends
mostly on the number of distinct origins and destinations rather than
the number of origin-destination \emph{pairs}, the algorithm is an
effective tool to build traffic models on comprehensive scales.

This paper is structured as follows: first, we introduce helpful definitions
and notation, review concepts we build on, and provide a clear definition
of our objective. Then we give an overview of REVC, before we decribe
each step in detail. After describing the algorithm, we present results
of numerical and empirical tests proving the algorithm's applicability
and efficiency in real-world problems. Finally, we discuss the test
results and the limitations and benefits of our approach.

\section{Algorithm}

\subsection{Preliminaries}

In this section, we specify our goal and introduce helpful notation
and concepts. First, we provide definitions and notation, which we
then use to characterize the routes we are seeking. Afterwards, we
recapitulate Dijkstra's algorithm and briefly describe the method
of reach based pruning, two basic concepts that our work builds on.

\subsubsection{Problem statement and notation}

Suppose we are given a graph $G=\left(V,E\right)$ that represents
a road network. The set of vertices $V$ models intersections of roads
as well as the start and end points of interest. The directed edges
$e\in E$ represent the roads of the road network and are assigned
non-negative weights $c_{e}$, denoting the costs for driving along
the roads. To ease notation, we will refer to the cost of an edge
or path as its \emph{length} without loss of generality. In practice,
other cost metrics, such as travel time, may be used. Our goal is
to find locally optimal paths between all combinations of origin locations
$s\in\origs\subseteq V$ and destination locations $t\in\dests\subseteq V$.

To specify the desired paths more precisely, we introduce convenient
notation and make some definitions:
\begin{lyxlist}{00.00.0000}
\item [{$P_{st}$}] is the shortest path from $s$ to $t$. 
\item [{$P^{uv}$}] is the subpath of $P$ from $u\in P$ to $v\in P$.
\item [{$\l P$}] is the length of the path $P$. That is, $\l P=\smo{e\in P}c_{e}$.
\item [{$\dd uv:=\l{P_{uv}}$}] is the length of the shortest path from
the vertex $u$ to the vertex $v$.
\item [{$\dd[P]uv:=\l{P^{uv}}$}] is the length of the subpath of $P$
from vertex $u$ to vertex $v$.
\end{lyxlist}
With this notation, we introduce the notion of \emph{single-via path}s.
\begin{defn}
A \emph{single-via path} (or short \emph{v-path}) $P_{svt}$ via a
vertex $v$ is the shortest path from a vertex $s$ to a vertex $t$
via $v$. We say $v$ \emph{represents} the single-via path $P_{svt}$
with respect to the origin-destination pair $\left(s,t\right)$.
\end{defn}
For simplicity, we assume that $P_{st}$ and $P_{svt}$ are always
uniquely defined. In practise, the paths are concatenations of shortest
paths found by algorithms outlined below, which are responsible for
breaking ties. 

We proceed with a precise definition of local optimality following
\citet{abraham_alternative_2013}. Generally speaking, a path is $T$-locally
optimal if each subpath of $P$ with a length of at most $T$ is a
shortest path. However, because paths are concatenations of discrete
elements, we need a more technical definition.
\begin{defn}
Consider a subpath $P'\subseteq P$ and let $P''\subset P'$ be $P'$
after removal of its end points. We say $P'$ is a \emph{$T$-significant}
subpath of $P$ if $\l{P''}<T$. A path $P$ is $T$-\emph{locally
optimal} if all its $T$-significant subpaths $P'$ are shortest paths.
We say $P$ is \emph{$\alpha$-relative locally optimal} if it is
$T$-locally optimal with $T=\alpha\cdot\l P$.\label{def:T-locally-optimal}
\end{defn}
We want to identify locally optimal paths between many origin and
destination locations. However, there may be an excessive number of
such paths. Therefore, we apply slightly stronger constraints on the
searched paths, which we will call \emph{admissible} below.
\begin{defn}
Let $\alpha\in(0,1]$ and $\beta\geq1$ be constants. A v-path $P_{svt}$
from vertex $s\in\origs$ to vertex $t\in\dests$ via vertex $v\in V$
is called \emph{admissible} if
\begin{enumerate}
\item $P_{svt}$ is $\alpha$-relative locally optimal.
\item $P_{svt}$ is longer than the shortest path by no more than factor
$\beta$, i.e. $\l{P_{svt}}\leq\beta\cdot\l{P_{st}}$.
\end{enumerate}
\end{defn}

\begin{objective}

The objective of this paper is to identify (close to) all admissible
single-via paths between each origin $s\in\origs$ and each destination
$t\in\dests$.

\end{objective}

\subsubsection{Dijkstra's algorithm}

Large parts of our algorithm are based on modifications of Dijkstra's
algorithm \citep{dijkstra_note_1959,dantzig_linear_1998}. Dijkstra's
algorithm is a frequently used method to find the shortest paths from
an origin $s$ to all other vertices in a graph with non-negative
edge weights. Though the algorithm is well-known to a large audience,
we briefly recapitulate the algorithm to establish some notation that
we will use later.
\begin{itemize}
\item In Dijkstra's algorithm, every vertex $v$ is assigned a specific
cost denoted $\cost v$. Eventually, this cost shall be equal to the
distance between the origin vertex $s$ and vertex $v$. Initially,
however, the cost of each vertex is $\infty$. An exception is the
origin $s$, for which the initial cost is $0$.
\item We say that a vertex $v$ is \emph{scanned} if we are certain that
$\cost v=\dd sv$. Furthermore, we say that a not yet scanned vertex
$v$ is \emph{labelled} if $\cost v<\infty$. All other vertices are
called \emph{unreached}. In line with our notion of scanned vertices,
we call edges $e=\left(u,v\right)$ scanned if we know that $e\in P_{sv}$
for some scanned vertex $v$.
\end{itemize}
Dijkstra's algorithm is outlined in Algorithm \ref{alg:Dijkstra's-algorithm}.
Initially, all vertices are in a container that allows us to determine
the least-cost vertex efficiently.  Dijkstra's algorithm consecutively
removes the least-cost vertex $v$ from the container and scans it.
That is, the algorithm iterates over $v$'s successors $w$ and updates
their costs if the distance from the origin $s$ to $w$ via $v$
is smaller than the current cost of $w$. In this case, $v$ is saved
as the \emph{parent} of $w$.

\begin{algorithm}[t]
\setstretch{1.35}
\SetKwProg{myproc}{}{}{}  
\While{container is not empty}{ 
	Take the vertex  with the lowest cost from the container and remove it\; 
	\myproc{Scan the vertex $v$:}{
		\ForAll{successors  of  $v$ that have not been scanned yet}{
			\myproc{Label $w$:}{ 
				\If{$\cost{w}<\cost{v}+c_{vw}$}{
					Set $\cost{w}:=\cost{v}+c_{vw}$ \tcp*[r]{$c_{vw}$ is the length of the edge from $v$ to $w$}
					Set $\parent{w}:=v$\;
				}
			}
		}
	}
}

\caption[Dijkstra's algorithm]{Dijkstra's algorithm. \label{alg:Dijkstra's-algorithm}}
\end{algorithm}

After execution of Dijkstra's algorithm, shortest paths can be reconstructed
by following the trace of the computed parent vertices, starting at
the destination vertex and ending at the origin. The edges $\left(\parent v,v\right)$
for all scanned vertices $v\in V$ form a \emph{shortest path tree}.
Hence, we call the procedure described above ``growing a shortest
path tree''.  The distance (measured in cost units) from the start
vertex to its farthest descendant is called the \emph{height} of the
shortest path tree. As we will see below, it can be beneficial to
stop the tree growth when the tree has reached a certain height.

When the shortest path between a specific pair of vertices $s$ and
$t$ is sought, the \emph{bidirectional Dijkstra algorithm} is more
efficient than the classic algorithm (compare Figures \ref{fig:Tree-growth}
(a) and (b)). The bidirectional Dijkstra algorithm grows two shortest
path trees: one in forward direction starting at the origin $s$ and
one in backward direction starting at the destination $t$. The trees
are grown simultaneously; i.e., the respective tree with smaller height
is grown until its height exceeds the other tree's height. The search
terminates if a vertex $v$ is included in both trees, i.e., scanned
from both directions. The shortest path is the concatenation of the
$s$-$v$ path in the first shortest path tree and the $v$-$t$ path
in the second tree.

\begin{figure}[t]
\includegraphics[scale=0.6]{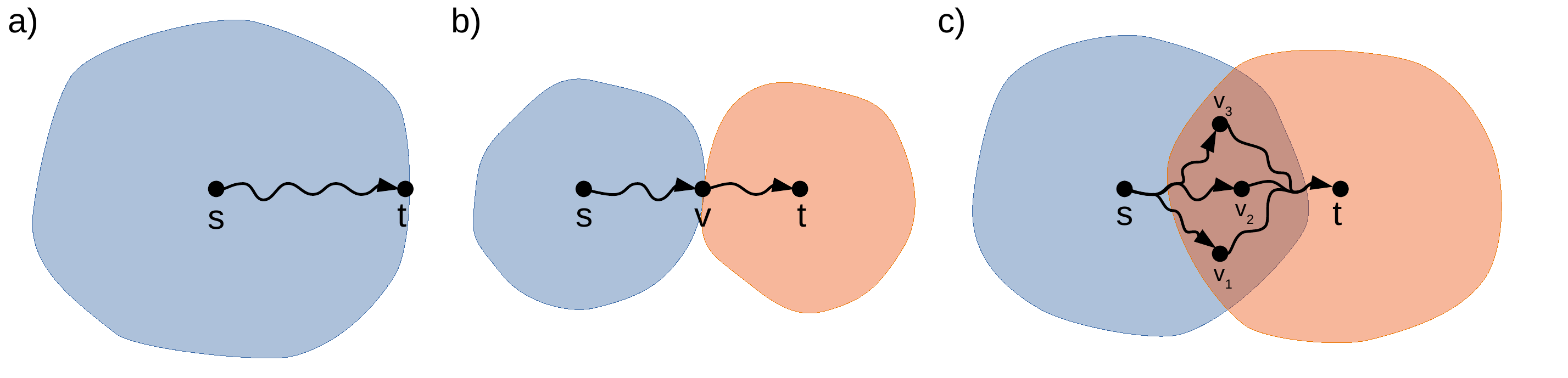}

\caption[Conceptual illustration of different path search algorithms]{Conceptual illustration of different path search algorithms for an
origin $s$ and a destination $t$. The shaded areas depict shortest
path trees. (a) Dijkstra's algorithm grows a single shortest path
tree around the origin until the destination is reached. (b) The bidirectional
Dijkstra algorithm grows a forward tree around the origin and a backward
tree around the destination until the two shortest path trees meet
at a vertex $v$. (c) Multiple v-paths can be constructed by growing
overlapping shortest path trees around origin and destination. \protect \\
Figures (a) and (b) are redrawn from \citet{kliemann_route_2016}.\label{fig:Tree-growth}}
\end{figure}

\subsubsection{Reach-based pruning}

Dijkstra's algorithm is not efficient enough to find shortest paths
in large networks within reasonable time. Therefore, multiple methods
have been developed to identify and prune vertices that cannot be
on the shortest path. One of these approaches is reach-based pruning
(RE; \citealp{raman_reach_2006}), which we introduce below.

Let us start by introducing the notion of a vertex's reach.
\begin{defn}
The \emph{reach} of a vertex $v$ is defined as \vspace{-0.5cm}
\begin{eqnarray}
\reach v & := & \maxo{u,w\in V\,:\,v\in P_{uw}}\left\{ \min\left(\dd uv,\,\dd vw\right)\right\} .
\end{eqnarray}
That is, if we consider all shortest paths that include $v$, split
each of these paths at $v$, and consider the shorter of the two ends,
then the reach of $v$ is the maximal length of these sections. The
reach of $v$ is high if $v$ is at the centre of a long shortest
path. Typically, vertices on highways have a high reach, since many
long shortest paths include highways.
\end{defn}
Disregarding vertices with small reaches can speed up shortest path
searches. Suppose we use the bidirectional Dijkstra algorithm to find
the shortest path between the vertices $s$ and $t$ and have already
grown shortest path trees with heights $h$. Let $v\in P_{st}$ be
a vertex that is located on the shortest path between $s$ and $t$
but has not been scanned yet. Then $\dd sv>h$ and $\dd vt>h$, since
$v$ would have been included in one of the shortest path trees otherwise.
Therefore, we know that $\reach v\geq\min\left(\dd sv,\,\dd vt\right)>h$.
Thus, when adding further vertices to our shortest path trees, we
can neglect all vertices with a reach less or equal to $h$. This
speeds up the shortest path search.

Computing the precise reaches of all vertices is expensive, as this
would require an extremely large number of shortest path queries.
However, \citet{raman_reach_2006} developed an algorithm to compute
upper bounds on vertices' reaches efficiently. These upper bounds
can be used in the same way as exact vertex reaches.

\subsection{Outline of the algorithm}

\begin{figure}[t]
\includegraphics[scale=0.6]{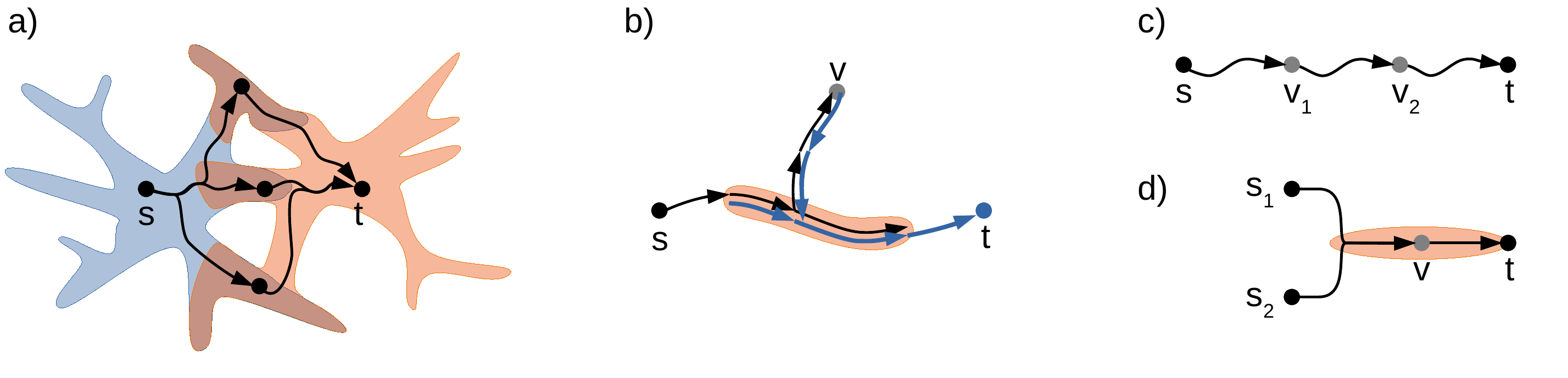}

\caption[Optimizations that REVC employs to efficiently identify admissible
v-paths]{Optimizations that REVC employs to efficiently identify admissible
v-paths between origin-destination pairs $\left(s,t\right)$. (a)
Shortest path trees (depicted as shaded areas) are grown to a tight
bound only and exclude low-reach vertices, which cannot be on long
locally optimal paths. (b) U-turn paths (e.g. $s\rightarrow v\rightarrow t$)
are excluded by requiring that an edge adjacent to the via vertex
is included both in the shortest path tree around the origin (black
arrows) and the shortest path tree around the destination (blue arrows).
Edges satisfying this constraint are highlighted with red background.
Note that arrows with different directions depict distinct edges.
(c) If v-paths via different vertices $v_{1}$ and $v_{2}$ are identical,
only one of these vertices is chosen to represent the path. (d) If
v-paths for different origin-destination pairs (here: $(s_{1,}t)$
and $(s_{2,}t)$) are represented by the same via vertex $v$ and
share a subpath (highlighted red), the local optimality of this section
is tested only once for all origin-destination pairs. \label{fig:Optimizations}}
\end{figure}

After specifying our goal and introducing necessary notation and concepts,
we can now proceed with an overview of our algorithm. The main idea
of REVC is (1) to grow shortest path trees in forward direction from
all origins and in backward direction from all destinations and (2)
to check the admissibility of the v-paths via the vertices that have
been scanned in both forward and backward direction (see Figure \ref{fig:Tree-growth}c).
For each vertex $v$ that is scanned both from an origin $s$ and
a destination $t$, the v-path $P_{svt}$ can be reconstructed easily
from the information contained in the shortest path trees. Therefore,
the only remaining step is to check whether $P_{svt}$ is admissible,
i.e. locally optimal and not much longer than the shortest path $P_{st}$.

As each vertex $v\in V$ could serve as via vertex for many origin-destination
combinations, checking the admissibility of all possible v-paths may
be infeasible. Therefore, it is important to identify and exclude
vertices that cannot represent admissible v-paths. The following observations
can be exploited: (1) v-paths via vertices that are very far from
an origin or destination cannot fulfill the length requirement. (2)
Some vertices represent intersections of minor roads, which can be
bypassed on close-by major roads. Thus, these vertices cannot be part
of locally optimal paths. (3) Some v-paths may include a u-turn at
the via vertex (see Figure \ref{fig:Optimizations}b). That is, travellers
driving on such a path would need to drive back and forth along the
same road. This is not locally optimal behaviour. (4) Some via vertices
may represent the same v-paths. That is, the v-paths corresponding
to distinct via vertices may be identical, and only one of these via
vertices needs to be considered.

\begin{figure}
\begin{centering}
\includegraphics[scale=0.8]{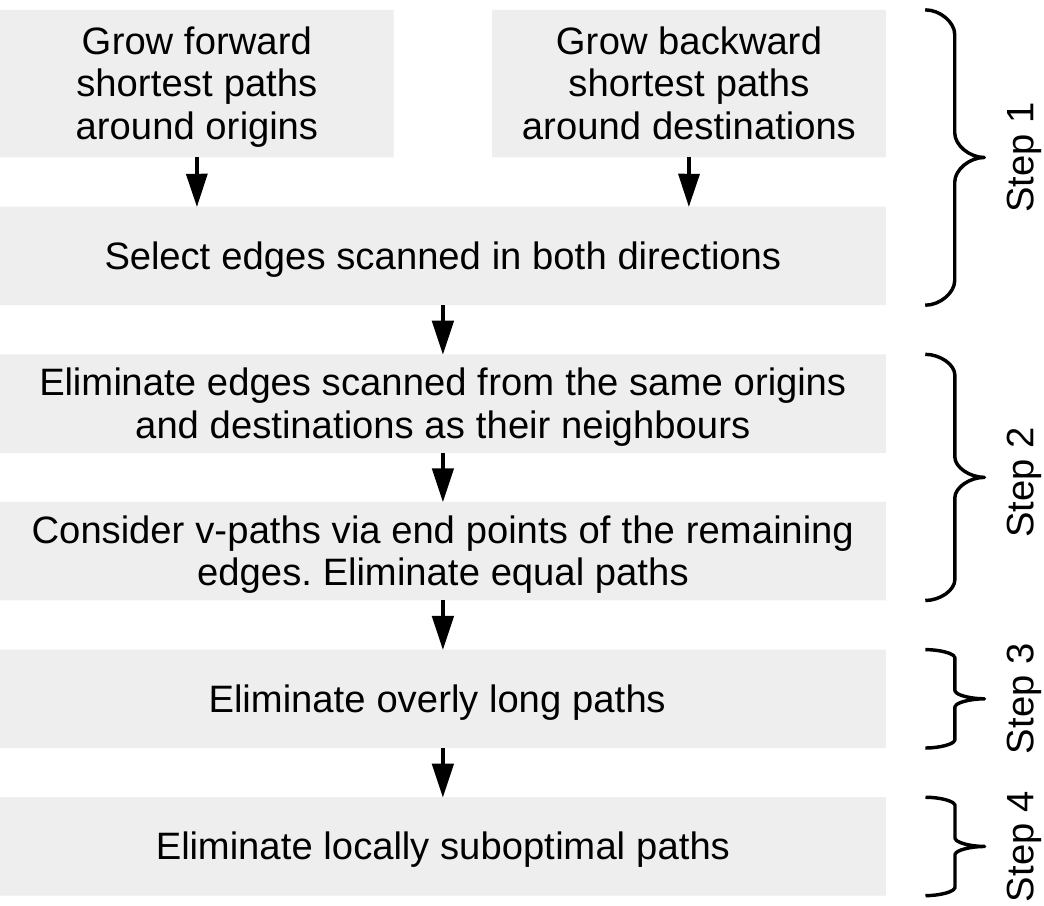}
\par\end{centering}
\caption[Overview of REVC]{Overview of REVC.\label{fig:overview}}
\end{figure}

Our algorithm REVC makes use of the observations listed above. (1)
When shortest path trees are grown around each origin and destination,
the trees are grown up to a tightly specified height only. That way,
many vertices that are too far off will not be scanned. (2) When the
shortest path trees are grown, reach based pruning is applied to exclude
vertices that are not on any sufficiently locally optimal path (see
Figure \ref{fig:Optimizations}a). (3) Instead of considering all
v-paths via vertices scanned in forward and backward direction, REVC
considers only v-paths in which an \emph{edge} adjacent to the via
vertex has been scanned forward and backward. This excludes paths
involving u-turns (see Figure \ref{fig:Optimizations}b). (4) Before
checking the admissibility of the remaining v-paths, the algorithm
ensures that each v-path is represented by one vertex only (see Figure
\ref{fig:Optimizations}c).

After these steps, REVC excludes v-paths that are exceedingly long
and checks which v-paths are sufficiently locally optimal. Testing
whether all v-paths $P_{svt}$ via a specific vertex $v$ are locally
optimal would be expensive if each origin-destination pair $\left(s,t\right)\in\origs\times\dests$
were considered individually. Therefore, REVC checks the admissibility
of many paths simultaneously, reusing earlier results and
applying approximations. That way, the algorithm becomes much more
efficient than individual pair-wise searches for admissible paths
(see Figure \ref{fig:Optimizations}d). In Figure \ref{fig:overview},
we provide an overview of REVC.

Before the actual algorithm can be started, some preparational work
and preprocessing is required. We will provide a detailed description
of the preprocessing procedure after introducing the algorithm in
detail.

\subsection{Step 1: Growing shortest path trees}

The algorithm REVC starts by growing forward shortest path trees out
of each origin and backward shortest path trees into each destination.
For each admissible v-path $P$, we need to scan at least one vertex
$v$ with $P=P_{svt}$ from both the origin $s$ and the destination
$t$. In addition, we want to scan one edge $e\in P$ adjacent to
$v$ from both directions if possible. These edges will be used to
exclude u-turn paths. For each vertex $v$ included in a shortest
path tree, we note $v$'s predecessor and height in the tree. Furthermore,
we memorize from which origins and destinations each edge has been
scanned. 

\subsubsection{Tree bound}

To save the work of scanning vertices inadmissibly far away from the
origins and destinations, we aim to stop the tree growth as soon as
possible.  We need to scan at least one vertex $v$ for each admissible
path $P_{svt}$ with a length $\l{P_{svt}}=\dd sv+\dd vt\leq\beta\cdot\l{P_{st}}$.
Since either of $\dd sv$ and $\dd vt$ could be arbitrarily small,
the algorithm REV by \citet{abraham_alternative_2013} grows the trees
up to a height of $\beta\cdot\l{P_{st}}$. Nevertheless, we can terminate
the search earlier if we take into account that we are searching for
locally optimal paths.

To derive a tighter tree bound, note that for an $\alpha$-relative
locally optimal path $P$, each subsection with length $\alpha\cdot\l P$
is a shortest path. This is in particular true for the subsection
$P'\subseteq P$ starting at the origin. Since $P'$ is a shortest
path, the end point $x_{s}$ of this subsection will be included in
the origin's shortest path tree. Therefore, it suffices to grow the
destination's shortest path tree until $x_{s}$ is reached, which
is closer to $t$ than $\beta\cdot\l{P_{st}}$. The same applies in
the reverse direction.

To specify the tree bound, define $x_{s}\in P$ more precisely to
be the first vertex that is farther away from the origin than $\alpha\cdot\l P$.
If this vertex is located in the second half of the path, change $x_{s}$
to be the last vertex in the first half of $P$. Choose $x_{t}$ accordingly
in relation to the destination. Our observations from above are formalized
in the following lemma and corollary, which we prove in Appendix \ref{sec:APX-Proofs}.
\begin{lem}
With $s$, $t$, $x_{s}$, $x_{t}$, and $P$ defined as above, there
is at least one vertex $v\in P$ with \label{lem:treeBoundLemma}
\vspace{-0.3cm}
\begin{enumerate}
\item $\dd[P]sv=\dd sv\leq\dd[P]s{x_{t}}$ \label{enu:treeBoundLemma1}
and
\item $\dd[P]vt=\dd vt\leq\dd[P]{x_{s}}t$.\label{enu:treeBoundLemma2}
\end{enumerate}
\end{lem}
\begin{cor}
\label{cor:treeBound}For each admissible v-path between an origin-destination
pair $\left(s,t\right)$, a via vertex will be scanned from both directions
if the shortest path trees are grown up to a height of 
\begin{eqnarray}
h_{\mathtext{max}} & := & \max\left\{ \left(1-\alpha\right)\beta\l{P_{st}},\,\frac{1}{2}\beta\l{P_{st}}\right\} .
\end{eqnarray}
\end{cor}
In Corollary \ref{cor:treeBound}, we consider a single origin-destination
pair. However, we want to identify admissible paths between multiple
origins and destinations and have to adjust the tree bound accordingly.
The tree around each origin and destination shall be large enough
to include via vertices for \emph{all} paths starting at the respective
endpoint. Hence, if we grow a tree out of origin $s$, we grow it
to a height of $\max\left\{ \left(1-\alpha\right)\beta M_{s},\,\frac{1}{2}\beta M_{s}\right\} $
with $M_{s}=\maxo{t\in\dests}\l{P_{st}}$. We proceed with destinations
similarly.

Note that the tree bounds above can only be determined if the shortest
distances between the origins and destinations are known. Though these
distances can be determined while the shortest path trees are grown,
 we will see in the next section that the shortest distances can
also be used to speed up the tree growth itself. Therefore, it is
beneficial to determine the shortest distances in a preprocessing
stage. This also makes it easy to grow the trees in parallel. 

\subsubsection{Pruning the trees}

The search for admissible paths can be significantly sped up if vertices
with small reach values are ignored when the shortest paths are grown.
Consider a vertex $v$ on an admissible $s$-$t$ path $P$. Let us
regard the subpath $P'$ that is centred at $v$ and has a length
just greater than $\alpha\cdot\l P$. Since $P$ is $\alpha$-relative
locally optimal, we know that $P'$ is a shortest path. Furthermore,
$P'$ is roughly split in half by $v$, unless $v$ is close to one
of the end points of $P$. Thus, 
\begin{eqnarray}
\reach v & \geq & \min\left\{ \frac{\alpha}{2}\l P,\dd sv,\dd vt\right\} \label{eq:reach-bound-0}
\end{eqnarray}
(see Lemma 5.1 in \citealp{abraham_alternative_2013}).

If we are growing the tree out of origin $s$, we can use (\ref{eq:reach-bound-0})
to prune the successors of vertices $v$ with $\reach v<\min\left\{ \frac{\alpha}{2}\l P,\dd sv\right\} $.
Pruning the successors but not $v$ itself ensures that at least one
vertex  per admissible path is scanned from both directions, even
if (\ref{eq:reach-bound-0}) is dominated by $\dd vt$.

Since $\l P$ is unknown when the shortest path trees are grown, the
length of $P$ must be bounded with known quantities. \citet{abraham_alternative_2013}
use the triangle inequality 
\begin{equation}
\l P\la{\geq}\dd sv+\dd vt\la{\geq}\cost v.
\end{equation}
However, we can also determine shortest distances before we search
admissible paths and exploit that $P\geq\dd st$ or, if we are considering
multiple origins and destinations, $\l P\geq L_{s}:=\mino{\tilde{t}\in\dests}\dd s{\tilde{t}}$.
Therefore, we may prune the successors of vertices $v$ with 
\begin{equation}
\reach v<\min\left\{ \cost v,\frac{\alpha}{2}\max\left\{ \cost v,L_{s}\right\} \right\} \label{eq:reach-bound-1}
\end{equation}
when we grow the shortest path tree out of origin $s$.

We can prune even more vertices if we grow the trees in forward and
backward direction in separate steps. The idea is to use data collected
in the first step to derive a sharper pruning bound for the second
step. Whether we grow the forward or the backward trees in the first
step depends on whether there are more destinations or more origins
to process. Below we assume without loss of generality that we consider
more destinations than origins, $\left|D\right|\geq\left|O\right|$.

\begin{algorithm}[p]
\setstretch{1.35}
\SetKwProg{myproc}{}{}{} 
\While{container is not empty}{ 
	Take the vertex $v$ with the lowest cost from the container and remove it\; 
	Mark edge leading to $v$ as visited from origin $s$\;
	Include $v$ in the shortest path tree\;
	\If{$d_{\mathtext{min}}\argparentheses v > \cost{v}$}{
		$d_{\mathtext{min}}\argparentheses v := \cost{v}$\;
	}
	\If{$\reach{v} \geq \min\left(cost(v),\frac{\alpha}{2}\max\left(\cost{v},L_{s}\right)\right)$}{
		Scan the vertex $v$\tcp*[r]{see Algorithm \ref{alg:Dijkstra's-algorithm}}
	}
}

\caption[Growing a forward shortest path tree out of an origin]{Growing a forward shortest path tree out of origin $s$. \label{alg:treeFromS}}
\end{algorithm}
\begin{algorithm}[p]
\setstretch{1.35}
\SetKwProg{myproc}{}{}{} 
\While{container is not empty}{ 
	Take the vertex $v$ with the lowest cost from the container and remove it\; 
	Mark edge leading to $v$ as visited from destination $t$\;
	\If{$\reach{v} \geq \min\left(\cost{v},\frac{\alpha}{2}\max\left(\cost{v},L_{t}\right)\right)$}{
		Include $v$ in the shortest path tree\;
		\myproc{Scan the vertex $v$ with early pruning:}{
		\ForAll{neighbors $w$ of  $v$ that have not been scanned yet}{
			$newCost:= \cost{v}+\dd v w$\;
			\If{$\reach{v} \geq \min\left(newCost, \frac{\alpha}{2}\max\left(newCost,L_{t}\right), d_{\mathtext{min}}\argparentheses{v} \right)$}{
				Label $w$\tcp*[r]{see Algorithm \ref{alg:Dijkstra's-algorithm}}
			}
		}
	}
	}
}

\caption[Growing a forward shortest path into a destination]{Growing a forward shortest path into destination $t$. \label{alg:treeToT}}
\end{algorithm}

We proceed as follows: we start by growing the forward trees out of
the origins. In this phase, we prune vertices' successors according
to inequality (\ref{eq:reach-bound-1}). After growing the forward
trees, we determine for each scanned vertex $v$ the distance $d_{\mathtext{min}}\argparentheses v:=\mino{s\in\origs;\,v\text{ scanned from }s}\dd sv$
to the closest origin it has been scanned from. If $v$ has not been
scanned, we set $d_{\mathtext{min}}\argparentheses v:=\infty$. Now
we grow the backward trees and use $d_{\mathtext{min}}\argparentheses v$
as a lower bound for $\dd sv$ for all origins $s\in\origs$. Hence,
we can prune all vertices with
\begin{equation}
\reach v<\min\left\{ \cost v,\frac{\alpha}{2}\max\left\{ \cost v,L_{t}\right\} ,d_{\mathtext{min}}\argparentheses v\right\} .\label{eq:reach-bound-2}
\end{equation}

In contrast to criterion (\ref{eq:reach-bound-1}), we can apply criterion
(\ref{eq:reach-bound-2}) directly to each vertex $v$ and not only
to its successors. This decreases the number of considered vertices.
We provide pseudo code for the tree growth procedures in Algorithms
\ref{alg:treeFromS} and \ref{alg:treeToT}.

\subsubsection{Determining potential via vertices}

With the shortest path trees, we can determine which vertices may
potentially represent admissible v-paths. Each vertex scanned in forward
and backward direction could be such a via vertex. However, since
some of the resulting paths could include u-turns, we consider the
scanned \emph{edges} rather than the vertices. This excludes paths
with u-turns (see Figure \ref{fig:via-vertices-via-edges}).

\begin{figure}
\begin{centering}
\includegraphics[scale=0.45]{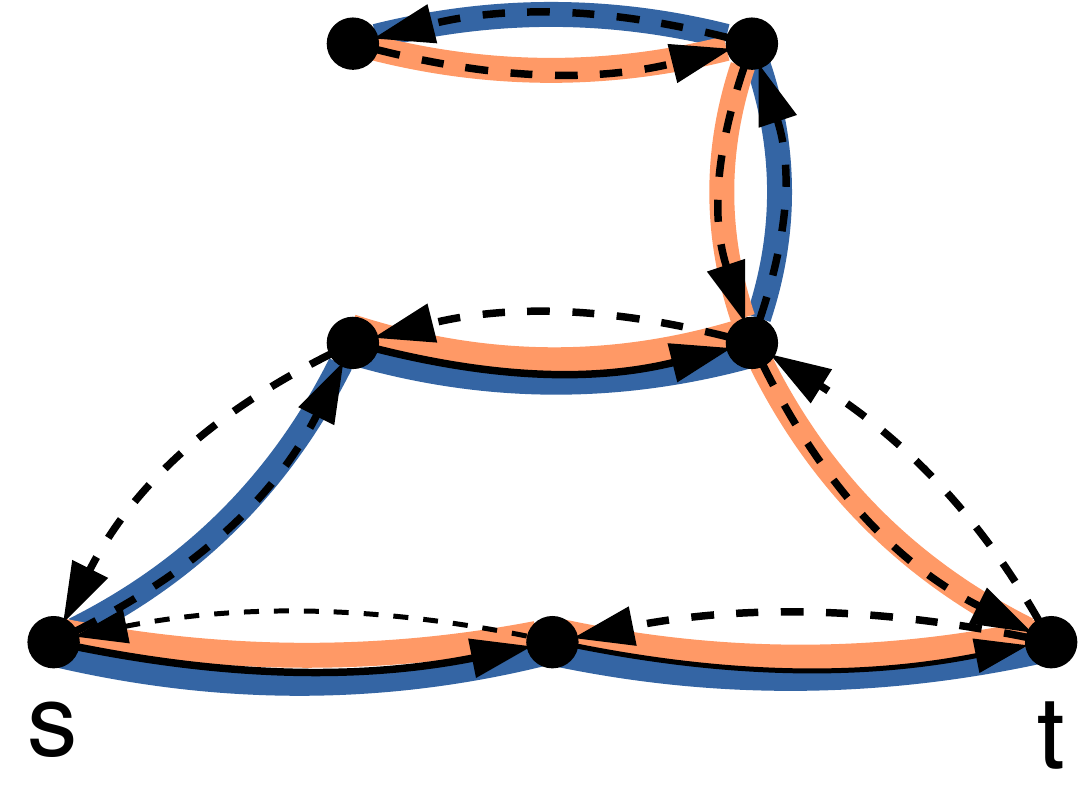}
\par\end{centering}
\caption[Advantages of considering via edges instead of via vertices]{Advantages of considering via edges instead of via vertices. Arrows
highlighted in dark blue depict the forward shortest path tree grown
from the origin $s$, and arrows highlighted in light red represent
the backward tree grown into the destination $t$. Edges that are
scanned from both directions are potential via edges and drawn as
solid black lines. The remaining edges are drawn as dashed black lines.
All vertices are scanned both from $s$ and $t$ and would therefore
considered potential via vertices. However, paths via the two topmost
vertices would require a u-turn. Restricting the focus on v-paths
via vertices adjacent to the solid lines excludes these u-turn paths.\label{fig:via-vertices-via-edges}}
\end{figure}
We proceed as follows: we determine for each scanned edge $e$ the
sets $O_{e}$ and $\dests_{e}$ of origins and destinations that $e$
has been scanned from. We discard all edges that have not been scanned
from at least one origin and one destination. Let $E_{\text{via}}$
be the resulting set of edges. The set of considered via vertices
$\vvia:=\left\{ v\in V\,|\,\exists w\in V:\left(v,w\right)\in\evia\right\} $
is given by the starting points of the edges in $\evia$. 

Note that though the procedure above eliminates paths with u-turns,
some admissible v-paths may be rejected as well. However, this issue
will rarely occur in realistic road networks, since the problem arises
only at specific merging points of very long edges. We provide details
in Appendix \ref{sec:APX-Admissible-paths-excluded-by-using-double-scanned-edges}.

\subsection{Step 2: Identifying vertices representing identical v-paths}

Some of the vertices in $\vvia$ may represent identical v-paths.
 Since we want to save the effort of checking the admissibility of
the same path multiple times and, similarly importantly, we do not
want to return multiple identical paths, we need to ensure that each
admissible path is represented by one via vertex only. 

To identify vertices representing identical paths, we have to compare
the v-paths corresponding to all $v\in\vvia$ for each origin-destination
pair. This requires $\order{\left|\vvia\right|\left|\origs\right|\left|\dests\right|}$
steps. However, for some vertices, identical paths can be identified
more quickly, as adjacent vertices typically represent similar sets
of v-paths. Therefore, we proceed in two steps: first, we reduce $\vvia$
by eliminating vertices whose via paths are also represented by their
respective neighbours, and second, we check which of the remaining
vertices represent identical v-paths. Below we describe the two steps
in greater detail.

\subsubsection{Eliminating vertices that represent the same v-paths as their neighbours\label{subsec:Eliminating-similar-paths-neighbours}}

The endpoints of an edge can be neglected as via vertices if the edge
has been scanned from the same origins and destinations as a neighbouring
edge. Consider for example an edge $\left(v,w\right)$ that has been
scanned from both an origin $s$ and a destination $t$. Then $P_{sw}=P_{svw}$
and $P_{vt}=P_{vwt}$. It follows that $v$ and $w$ represent the
same v-path with respect to $\left(s,t\right)$: $P_{svt}=P_{swt}$.
Now consider an adjacent edge $\left(u,v\right)$ that has been scanned
from $s$ and $t$ as well. Clearly, it is $P_{sut}=P_{svt}$ and
$P_{svt}=P_{swt}$, which implies that the v-paths via $u$, $v$,
and $w$ are identical. Therefore, only one of these vertices has
to be considered.

To introduce an algorithm that efficiently detects such configurations,
let $O_{e}$ be the set of origins and $D_{e}$ the set of destinations
that edge $e$ has been scanned from.  For each edge $e\in\evia$,
we check whether one directly preceding edge $e'\in\evia$ has been
scanned from a superset of origins and destinations, i.e. $\origs_{e}\subseteq\origs_{e'}$
and $\dests_{e}\subseteq\dests_{e'}$. If such an edge exists and
one of the set inequalities holds strictly, i.e. $\origs_{e}\subset\origs_{e'}$
or $\dests_{e}\subset\dests_{e'}$, we may disregard edge $e$, as
all  v-paths via $e$ are also v-paths via $e'$.

\begin{algorithm}[tp]
\setstretch{1.35}
\SetKwProg{myproc}{}{}{} 
\SetKwFunction{FPred}{has\_superior\_predecessor}
\SetKwFunction{FSuc}{has\_superior\_successor}
\SetKw{And}{and}
\SetKw{Not}{not}
\SetKwProg{Fn}{Function}{:}{}
\Fn{\FPred{$e$}}{
	Remove $e$ from $\evia$\;
	\ForAll{directly preceding edges $e'$ of $e$}{
		\If{$\origs_{e}\subseteq\origs_{e'}$ and $\dests_{e}\subseteq\dests_{e'}$}{
			\eIf{$\origs_{e} = \origs_{e'}$ and $\dests_{e} = \dests_{e'}$}{
				\Return \FPred{$e'$}
			}{
				\Return True\;
			}
			\Return False\;
		}
	}
}
\Fn{\FSuc{$e$}}{
	Remove $e$ from $\evia$\;
	\ForAll{directly succeeding edges $e'$ of $e$}{
		\If{$\origs_{e}\subseteq\origs_{e'}$ and $\dests_{e}\subseteq\dests_{e'}$}{
			\eIf{$\origs_{e} = \origs_{e'}$ and $\dests_{e} = \dests_{e'}$}{
				\Return \FSuc{$e'$}
			}{
				\Return True\;
			}
			\Return False\;
		}
	}
}
$\evia' := \emptyset$\;
\While{$\evia \neq \emptyset$}{
	Set $e := \text{ next entry in } \evia'$\;
	\If{\Not \FPred{$e$} \And \Not \FSuc{$e$}}{
		Add $e$ to $\evia'$\;
	}
}
$\evia := \evia'$\;

\caption[Eliminating vertices that represent the same v-paths as their neighbours]{Eliminating vertices that represent the same v-paths as their neighbours.
\label{alg:Eliminating-neighbours}}
\end{algorithm}

Things become more complicated if $\origs_{e}=\origs_{e'}$ and $\dests_{e}=\dests_{e'}$,
as we may either reject $e$, $e'$, or both edges. The latter case
may occur if $e'$ has another directly preceding edge $e''\in\evia$
with $\origs_{e'}\subseteq\origs_{e''}$ and $\dests_{e'}\subseteq\dests_{e''}$.
If one of these inequalities is strict, we disregard both $e$ and
$e'$. Otherwise, we continue traversing the edges in $\evia$ until
either (1) an edge is found whose origin and destination sets supersede
the sets of all previous edges or (2) no further predecessor with
sufficiently large origin and destination sets is found. In the second
case, we may disregard all traversed edges but $e$. We apply the
same approach to the successors of $e$ and repeat this procedure
until all edges in $\evia$ have been processed.

The updated set $\vvia$ of via vertices consists of the starting
vertices of the edges in the reduced edge set $\evia$. We provide
pseudo code for the outlined algorithm in Algorithm \ref{alg:Eliminating-neighbours}.
An efficient implementation may compare the origin and destination
sets of the edges in $\evia$ before the traverse is started. This
makes it easy to implement the most expensive parts of the algorithm
in parallel.

\subsubsection{Identifying remaining identical v-paths\label{subsec:Eliminating-identical-via-paths}}

The method outlined above identifies vertices that represent the same
v-paths as their neighbours. However, two vertices may represent the
same v-path with respect to one origin-destination pair but different
v-paths with respect to another origin-destination pair. Consequently,
these vertices could not be rejected in the step described above,
and a second procedure is required to eliminate the remaining identical
v-paths. 

We identify the remaining identical v-paths by comparing path lengths.
To this end, we assume that $P_{svt}=P_{swt}$ if and only if $\l{P_{svt}}=\l{P_{swt}}$.
Though it can happen that distinct paths have the same length, this
case is usually not of greater concern in practical applications.
The issue can be reduced by introducing a small random perturbation
for the lengths of edges. We examine this limitation further in
the discussion section.

With the above assumption, identical paths can be identified efficiently.
Since for each origin-destination pair $\left(s,t\right)$ and each
potential via vertex $v\in\vvia$ the distances $\dd sv$ and $\dd vt$
are known, the v-path lengths can be computed easily. For each origin-destination
pair, a comparison of the lengths of the v-paths corresponding to
all $v\in\vvia$ can be conducted in linear average time via hash
maps. Note that the path lengths must be compared with an appropriate
tolerance for machine imprecision.

In later steps it will be of benefit if most v-paths are represented
by a small set of via vertices. If there are multiple vertices representing
the same v-paths, we therefore choose the via vertex $v$ that has
been scanned from the most origin-destination combinations $\origs_{v}\times\dests_{v}$.
This makes it easier to reuse partial results when we check whether
the v-paths are locally optimal.

\subsection{Step 3: Excluding long paths}

Before we check whether paths are sufficiently locally optimal, we
exclude the paths that exceed the length allowance. That is, we disregard
all paths $P_{svt}$ with $\l{P_{svt}}>\beta\cdot\l{P_{st}}$ with
origin-destination pairs $\left(s,t\right)$ and via vertices $v\in\vvia$.
Since this step involves a simple comparison only, it is computationally
cheaper than identifying identical paths. Therefore, it is efficient
to conduct this step just before identical paths are eliminated (section
\ref{subsec:Eliminating-identical-via-paths}). This also reduces
the memory required to store potentially admissible combinations $\left(s,v,t\right)$
of origin-destination pairs and via vertices.

\subsection{Step 4: Excluding locally suboptimal paths}

The most challenging part of the search for admissible paths is to
check whether paths are sufficiently locally optimal. To test whether
a subpath is optimal, we need to find the shortest alternative, which
is computationally costly. Therefore, we apply an approximation to
limit the number of necessary shortest path queries.

Our method generalizes the approximate local optimality test by \citet{abraham_alternative_2013}.
 They noted that v-paths are concatenations of two optimal paths.
Hence, v-paths are locally optimal everywhere except in a neighbourhood
of the via vertex. More precisely, a v-path $P_{svt}$ from $s$ to
$t$ via $v$ is guaranteed to be $T$-locally optimal everywhere
except in the section that begins $T$ distance units before $v$
and ends $T$ distance units after $v$. Therefore, \citet{abraham_alternative_2013}
suggest to perform a shortest path query between the end points $x$
and $y$ of this section to check whether it is optimal. \citet{abraham_alternative_2013}
call this procedure the T-test.

The T-test does not return false positives. That is, a path that is
not $T$-locally optimal will never be misclassified as locally optimal.
However, the T-test may return false negatives: paths that are $T$-locally
optimal but not $2T$-locally optimal may be rejected. In modelling
applications, a more precise local optimality test may be desired.

It is possible to increase the precision of the T-test. Instead of
checking whether the whole potentially suboptimal subpath is optimal,
we may test multiple subsections to gain a higher accuracy. While
this procedure ensures that fewer admissible paths are falsely rejected,
the gain in accuracy comes with an increase in computational cost.
Therefore, it is desirable to use the results of earlier local optimality
checks to test the admissibility of other paths.

There are two situations in which local optimality results can be
reused. First, if a subsection of a path is found to be suboptimal,
other paths that include this section can be rejected as well. Second,
if a subpath of a path is found to be locally optimal, other paths
including this subpath may be classified as locally optimal as well.
That way, many paths can be processed all at once.

When reusing partial results, it is important to note that even though
we require all paths to be $\alpha$-relative locally optimal, the
\emph{absolute} lengths of the subsections that need to be optimal
 depend on how long the considered paths are. Therefore, paths must
be considered in an order dependent on their lengths. We provide details
below.

\subsubsection{Preparation\label{subsec:Preparation}}

Before we can start testing whether the remaining v-paths are locally
optimal, a preparation step is needed to identify the subpaths that
may be suboptimal and thus need to be assessed more closely. To reuse
partial results efficiently, we furthermore need to determine subsections
that different paths have in common. We describe the preparation procedure
below.

\sloppy We start by introducing helpful notation. Suppose we want
to test whether the v-paths via vertex $v$ are locally optimal. Let
$\tilde{\origs}:=\left\{ s\in\origs\,|\,\exists t\in\dests\,:\,\l{P_{svt}}\leq\beta\cdot\l{P_{st}}\right\} $
be the origins for which at least one destination can be reached via
$v$ without violating the length constraint. Let $\tilde{\dests}$
be defined accordingly for the destinations. Define $\tilde{\dests}_{s}:=\left\{ t\in\tilde{\dests}\,|\,\l{P_{svt}}\leq\beta\cdot\l{P_{st}}\right\} $
as the set of destinations that can be reached from the origin $s$
via $v$ without violating the length constraint. 

In the preparation step, we determine for each origin $s\in\tilde{\origs}$
the destination $t_{s}:=\argmaxo{t\in\tilde{\dests}_{s}}\l{P_{svt}}$
for which the potentially suboptimal section is longest. Furthermore,
we search for the vertex $x_{s}:=\argmino{\tilde{x}\in P_{sv};\,{{\dd{\tilde{x}}v}}\geq\alpha\l{P_{svt_{s}}}}\dd{\tilde{x}}v$,
which is the last vertex on $P_{sv}$ with $\dd{x_{s}}v\geq\alpha\cdot\l{P_{svt_{s}}}$,
and we determine $x_{t}$ defined accordingly. Now we fill the arrays

\begin{align}
A_{us} & :=\begin{cases}
\text{True} & \text{if }u\in P_{sv}\\
\text{False} & \text{else,}
\end{cases} & A_{ut} & :=\begin{cases}
\text{True} & \text{if }u\in P_{vt}\\
\text{False} & \text{else}
\end{cases}\label{eq:A-def}
\end{align}
for all vertices $u\in P_{x_{s}v}$ and $u\in P_{vx_{t}}$, respectively.

The information saved in the shortest path trees are suitable to find
paths from scanned vertices to the origins and destinations. However,
the trees contain no information on the reverse paths starting at
the end points. That is, while it is easy to find the backward shortest
path from $v$ to $x_{s}$, it is hard to follow the path in the opposite
direction starting at $x_{s}$. We gather the necessary information
in the preparation step: for each origin $s\in\tilde{\origs}$, we
save the successors of each relevant vertex $u\in P_{sv}$.

In Algorithm \ref{alg:preparation-LOTest}, we provide pseudo code
for the described procedures. The pseudo-code considers the origins
only. The algorithm for the destinations is similar. The preparation
phase ends with sorting all origin-destination pairs with respect
to the lengths of the respective v-paths via $v$.

\begin{algorithm}[t]
\setstretch{1.35}
\SetKwProg{myproc}{}{}{} 
\SetKw{wForAll}{for all}
\SetKw{Not}{not}
\ForEach{destination $s\in\tilde{\origs}$}{ 
	$t_{s} := \argmaxo{t\in \tilde{\dests}_{s}}\left(\dd s v + \dd v t \right)$\;
	$u := \parent[s]v$\;
	$\successor[s] u := v$\;
	$\text{stop} := \text{False}$\;
	\While{\Not stop}{ 
		\If{$u \notin A$}{
			Initialize $A_{u\tilde{s}} := \text{False}$ \wForAll {$\tilde{s}\in\tilde{\origs}$}\;
		}
		$A_{us} := \text{True}$\;
		$\successor[s]{\parent u} := u$\;
		\eIf{$\dd v u > \alpha\left(\dd s v + \dd v {t_{s}} \right)$}{
			$\text{stop} := \text{True}$\;
		}{
			$u := \parent u$\;
		}	
	}
	
}

\caption[Filling the array $A$ and finding successors]{Filling the array $A$ for the origins and finding successors. The
algorithm for the destinations is similar. \label{alg:preparation-LOTest}}
\end{algorithm}

\subsubsection{Testing local optimality for one origin-destination pair}

We use an approximation approach with flexible precision to check
whether paths are locally optimal. For a parameter $\delta\in\left[1,2\right]$,
we call this procedure the $\T[\delta]$-test. The parameter $\delta$
is a measure for the test's precision.

To outline the $\T[\delta]$-test, let us consider a v-path $P:=P_{svt}$
from $s$ to $t$ via the vertex $v$. Let $S_{s}:=\left\{ u\in P_{sv}\,|\,d(u,v)<T\right\} $
be the set of vertices that are on the path $P_{sv}$ and have a distance
less than $T$ to the vertex $v$. Furthermore, add to $S_{s}$ the
vertex $x:=\argmino{\tilde{x}\in P_{sv};\,{{\dd{\tilde{x}}v}}\geq T}\dd{\tilde{x}}v$
that is closest to $v$ but has $\dd xv\geq T$ if such a vertex exists.
Choose $S_{t}$ accordingly with respect to the destination vertex
$t$. Let $\partner[t]u{\tau}:=\argmino{\tilde{w}\in S_{t};\,{{\dd[P]u{\tilde{w}}}}\geq\tau}\dd[P]u{\tilde{w}}$
for $u\in S_{s}$ be the vertex $w\in S_{t}$ that is closest to $u$
but has $\dd[P]uw\geq\tau$. If no such vertex exists in $S_{t}$,
set $\partner[t]u{\tau}=y:=\argmaxo{\tilde{w}\in S_{t}}\dd[P]u{\tilde{w}}$.
Define accordingly $\partner[s]w{\tau}$ for $w\in S_{t}$ as the
vertex $u\in S_{s}$ that is closest to $w$ but has $\dd[P]uw\geq\tau$.

The $\T[\delta]$-test proceeds as follows: the algorithm starts at
the vertex $u_{1}:=x$ and checks whether the subpath $P^{u_{1}w_{1}}$
between $u_{1}$ and $w_{1}:=\partner[t]{u_{1}}{\delta T}$ is a shortest
path. If so, the algorithm progresses searching $u_{2}:=\partner[s]{u_{1}}T$
in backward direction and repeats the steps formerly applied to $u_{1}$
now with $u_{2}$. This procedure repeats until $u_{n}=v$ for some
$n\in\N$. If all the shortest path queries yield subpaths of $P$,
the path is deemed approximately $T$-locally optimal. Otherwise,
it is classified as not locally optimal. We depict the algorithm in
Figure \ref{fig:T_delta-test} and provide pseudo-code in Algorithm
\ref{alg:T-test}.

\begin{figure}[t]
\includegraphics[width=1.01\textwidth]{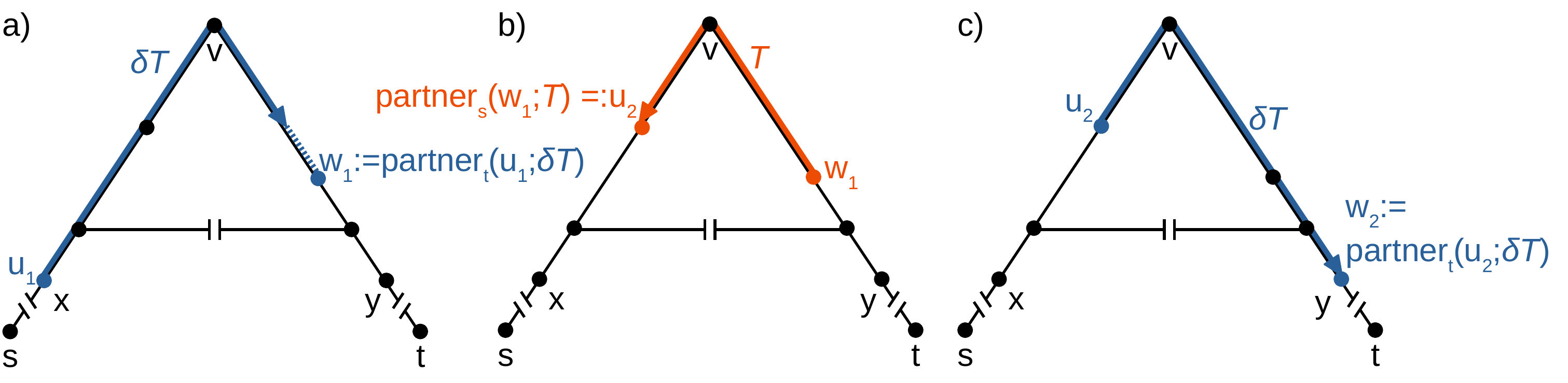}

\caption[{$\T[\delta]$-test}]{$\protect\T[\delta]$-test with $\delta=1.4$. The three subfigures
depict the steps of the $\protect\T[\delta]$-test for a path $P_{svt}$
connecting origin-destination pair $\left(s,t\right)$ via vertex
$v$. The vertices $x$ and $y$ are the end points of the potentially
locally suboptimal section. The edge lengths are given by the Euclidean
distance except for the edges with an indicated gap. (a) In a first
step, the test determines the vertex $w_{1}$ that is at least $\delta T$
units along the path away from $u_{1}:=x$ (the distance is depicted
as blue arrow). (b) If the shortest path query between $u_{1}$ and
$w_{1}$ indicates that the subsection $P_{svt}^{u_{1}w_{1}}$ is
optimal, the test continues by determining the first vertex $u_{2}$
that is at least $T$ units away from $w_{1}$ in backwards direction.
(c) From $u_{2}$, the algorithm searches the vertex $w_{2}$ that
is at least $\delta T$ units along the path beyond $u_{2}$ and conducts
a shortest path query between $u_{2}$ and $w_{2}$. If all the shortest
path queries yield subpaths of $P_{svt}$, the path is deemed approximately
$T$-locally optimal. Note that a $\protect\T[2]$-test would have
misclassified the path as not locally optimal, provided the shortest
path from $x$ to $y$ includes the horizontal edge. \label{fig:T_delta-test}}
\end{figure}

\begin{algorithm}[t]
\setstretch{1.35}
\SetKwProg{myproc}{}{}{}  

Search for the vertex $x \in S_s$ with maximal distance to $v$\;
Set $u := x$\;
Set $w := v$\;
\While{$ u \neq v$ and $ w \neq y$}{ 
	Set $w' := \partner[t]{u}{\delta T}$\;
	\eIf{$w = w'$}{
		Set $w :=$ next farthest vertex to $v$ in $S_t$\;
	}{
		Set $w := w'$\;
	}
	\myproc{Check whether the $u$-$w$ subpath is optimal}{
		\If{$d(u,w) < \dd u v + \dd v w$}{
			\Return "Not locally optimal"
		}
	}	
	Set $u' := \partner[s]{w}{T}$\;
	\eIf{$u = u'$}{
		Set $u :=$ next closest vertex to $v$ in $S_s$\;
	}{
		Set $u := u'$\;
	}
}
\Return "Locally optimal"

\caption[{$\T[\delta]$-test}]{$\protect\T[\delta]$-test. \label{alg:T-test}}
\end{algorithm}

Similar to the T-test, the $\T[\delta]$ test does not return false
positives. However, paths that are $T$-locally optimal but not $\delta T$-locally
optimal might be rejected. Hence, the $\T[1]$-test is exact,
whereas the ``classical'' T-test by \citet{abraham_alternative_2013}
is the $\T[2]$-test. An increase in precision comes with a computational
cost. The $\T[\delta]$-test requires at most $2\left\lceil \frac{1}{\delta-1}\right\rceil $
shortest path queries if $\delta>1$. However, query numbers around
$\frac{1}{\delta-1}$ are more common. Either way, the number of required
queries is bounded by a constant independent of the graph, unless
$\delta=1$.

\subsubsection{Using test results to check local optimality for multiple origin-destination
pairs\label{subsec:LO-results-for-many-pairs}}

The $\T[\delta]$-test is a suitable procedure to check whether a
single v-path is locally optimal. However, if many v-paths shall be
tested, the required number of shortest path queries may exceed a
feasible limit. Therefore, we show below how negative test results
can be used to reject multiple paths at once. Afterwards we describe
a method to use positive test results to classify many paths as locally
optimal.

\begin{figure}[t]
\begin{centering}
\includegraphics[scale=0.6]{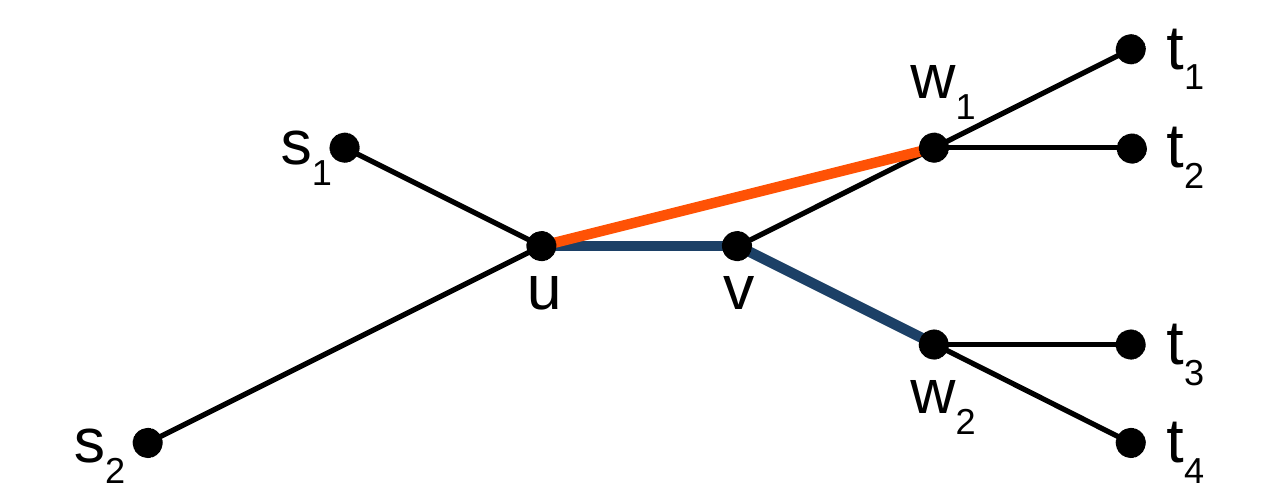}
\par\end{centering}
\caption[Accepting and rejecting multiple paths at once]{Accepting and rejecting multiple paths at once. Suppose we want to
check the admissibility of the paths from the origins $s_{i}$ to
the destinations $t_{j}$ via the vertex $v$. Suppose that we start
with the path $P_{s_{1}vt_{2}}$ from $s_{1}$ to $t_{2}$ via $v$
and find that the subsection $P_{uvw_{1}}$ is not optimal, because
there is a shorter path (light orange) from $u$ to $w_{1}$. Then we know
that the paths $P_{s_{1}vt_{1}}$, $P_{s_{2}vt_{1}}$, and $P_{s_{2}vt_{2}}$
are not sufficiently locally optimal, either. Now suppose we continue
with the pair $\left(s_{1},t_{3}\right)$ and find that $P_{s_{1}vt_{3}}$
is locally optimal because the section $P_{uvw_{2}}$ (dark blue) is optimal.
Since $P_{s_{1}vt_{4}}$ includes this subsection, too, and is not
much longer than $P_{s_{1}vt_{3}}$, we can deduce that $P_{s_{1}vt_{4}}$
is approximately locally optimal as well. \label{fig:Accepting-and-rejecting-multiple}}
\end{figure}

\paragraph{Rejecting paths}

Suppose that in order to test whether $P_{svt}$ is admissible, we
have checked whether the subpath $P_{svt}^{uw}$ between some vertices
$u$ and $w$ is a shortest path, and suppose we have obtained a negative
result, i.e. we have found that $\dd uw<\dd uv+\dd vw$. We can not only conclude
that the path $P_{svt}$ is not locally optimal but also reject other
v-paths that include the subpath $P_{svt}^{uw}$ (see Figure \ref{fig:Accepting-and-rejecting-multiple}).

To see which paths can be rejected, let $\Omega_{u}:=\left\{ \tilde{s}\in\origs\,|\,\dd{\tilde{s}}v=\l{P_{\tilde{s}uv}}\right\} $
be the set of origins for which $u$ is on the shortest path to $v$
and define $\Delta_{w}:=\left\{ \tilde{t}\in\dests\,|\,\dd v{\tilde{t}}=\l{P_{vw\tilde{t}}}\right\} $
accordingly for the destinations. Let furthermore $\mathcal{P}:=\left\{ \left(s,t\right)\in\tilde{\origs}\times\tilde{\dests}\,|\,\l{P_{svt}}\leq\beta\cdot\l{P_{st}}\right\} $
be the set of all origin-destination pairs with a potentially admissible
v-path via $v$, and let $\mathcal{P}_{uw}:=\mathcal{P}\cap\left(\Omega_{u}\times\Delta_{w}\right)$
denote the respective set of origin-destination pairs for which the
v-path via $v$ also includes $u$ and $w$. The following lemma shows
which paths can be rejected as approximately inadmissible.
\begin{lem}
Suppose the $\T[\delta]$-test is applied to check whether a path
$P_{svt}$ is $\alpha$-relative locally optimal and that the test
fails, because $\dd uw<\dd uv+\dd vw$ for some vertices $u$ and
$w$. Then, for each pair $\left(\tilde{s},\tilde{t}\right)\in\mathcal{P}_{uw}$
with $P_{\tilde{s}v\tilde{t}}\geq\l{P_{svt}}$, the v-path $P_{\tilde{s}v\tilde{t}}$
is not relative locally optimal with a factor higher than $\alpha_{\tilde{s}v\tilde{t}}<\frac{\l{P_{xvy}}}{\l{P_{svt}}}\leq\alpha\delta$,
whereby $x$ and $y$ are the neighbours of $u$ and $w$ in direction
of $v$, respectively. \label{lem:LO-rejection}
\end{lem}
\begin{proof}
By construction of $\mathcal{P}_{uw}$, it is $P_{xvy}\subseteq P_{\tilde{s}v\tilde{t}}$
for any origin-destination pair $\left(\tilde{s},\tilde{t}\right)\in\mathcal{P}_{uw}$.
Therefore, $P_{\tilde{s}v\tilde{t}}$ is at most $T$-locally optimal
with $T<\l{P_{xvy}}$. Hence, the local optimality factor $\alpha_{\tilde{s}v\tilde{t}}$
for $P_{\tilde{s}v\tilde{t}}$ satisfies
\begin{equation}
\alpha_{\tilde{s}v\tilde{t}}=\frac{T}{\l{P_{\tilde{s}v\tilde{t}}}}<\frac{\l{P_{xvy}}}{\l{P_{\tilde{s}v\tilde{t}}}}\leq\frac{\l{P_{xvy}}}{\l{P_{svt}}}\leq\frac{\alpha\delta\l{P_{svt}}}{\l{P_{svt}}}=\alpha\delta.
\end{equation}
\end{proof}
Following Lemma \ref{lem:LO-rejection}, we can reject all pairs $\left(\tilde{s},\tilde{t}\right)\in\mathcal{P}_{uw}$
with $P_{\tilde{s}v\tilde{t}}\geq\l{P_{svt}}$. The origin-destination
pairs in question can be determined by considering the array $A$
constructed in the preparation phase (equation (\ref{eq:A-def})).
Let $\tilde{A}_{u}:=\left\{ s\in\tilde{\origs}\,|\,A_{us}=\text{True}\right\} $
and $\tilde{A}_{w}:=\left\{ t\in\tilde{\dests}\,|\,A_{wt}=\text{True}\right\} $.
Then, $\mathcal{A}_{uw}:=\tilde{A}_{u}\times\tilde{A}_{w}\subseteq\mathcal{P}_{uw}$,
and $\mathcal{P}_{uw}\backslash\mathcal{A}_{uw}$ contains only pairs
$\left(\tilde{s},\tilde{t}\right)$ with $\l{P_{\tilde{s}v\tilde{t}}}<\l{P_{svt}}$.
It follows that all pairs $\left(\tilde{s},\tilde{t}\right)\in\mathcal{P}_{uw}$
with $P_{\tilde{s}v\tilde{t}}\geq\l{P_{svt}}$ are also in $\mathcal{A}_{uw}$. 

As $\mathcal{A}_{uw}$ may also contain pairs $\left(\tilde{s},\tilde{t}\right)$
with $\l{P_{\tilde{s}v\tilde{t}}}<\l{P_{svt}}$, we process the origin-destination
pairs in the order of increasing via-path length. Then the pairs $\left(\tilde{s},\tilde{t}\right)\in\mathcal{A}_{uw}$
with $\l{P_{\tilde{s}v\tilde{t}}}<\l{P_{svt}}$ will be processed
before $\left(s,t\right)$. If we label these pairs as ``processed''
and exclude them from $\mathcal{A}_{uw}$, then we can reject all
remaining pairs in $\mathcal{A}_{uw}$.

\paragraph{Accepting paths}

The procedure outlined in the previous section allows us to reject
many inadmissible paths with a single shortest distance query. However,
the procedure may yield limited performance gain if many of the considered
paths are admissible. Therefore, we introduce a second relaxation
of our local optimality condition: we classify paths as (approximately)
admissible if they are $\left(\alpha\gamma\right)$-relative locally
optimal with some constant $\gamma\in(0,1]$.

To see how this relaxation can be exploited, suppose that we are considering
an origin-destination pair $\left(s,t\right)$ and that we have already
confirmed that the path $P_{svt}$ is $\alpha$-relative locally optimal.
Let $x:=\argmino{\tilde{x}\in P_{sv};\,{{\dd{\tilde{x}}v}}\geq\alpha\l{P_{svt}}}\dd{\tilde{x}}v$
be the last vertex on $P_{sv}$ with a distance to $v$ of at least
$\alpha\cdot\l{P_{svt}}$. Let $y:=\argmino{\tilde{y}\in P_{vt};\,{{\dd v{\tilde{y}}}}\geq\alpha\l{P_{svt}}}\dd v{\tilde{y}}$
be defined accordingly for the destination branch. During the $\T[\delta]$-test
we have ensured that the section $P_{xvy}$ is approximately $T$-locally
optimal with $T=\alpha\cdot\l{P_{svt}}$.

In the lemma below, we identify the paths that can be classified
as approximately admissible after a successful $\T[\delta]$-test.
In line with the notation in the previous section, let $\Omega_{x}:=\left\{ \hat{s}\in\origs\,|\,\dd{\hat{s}}v=l\left(P_{\hat{s}xv}\right)\right\} $,
$\Delta_{y}:=\left\{ \hat{t}\in\dests\,|\,\dd v{\hat{t}}=\l{P_{vy\hat{t}}}\right\} $,
and $\mathcal{P}_{xy}:=\mathcal{P}\cap\left(\Omega_{x}\times\Delta_{y}\right)$. 
\begin{lem}
Let $\left(s,t\right)\in\mathcal{P}$ be an origin-destination pair.
If the $\T[\delta]$-test applied to $P_{svt}$ considered the vertices
on $P_{xvy}\subseteq P_{svt}$ and confirmed that the path is $\alpha$-relative
locally optimal, then all paths $P_{\tilde{s}v\tilde{t}}$ with $\left(\tilde{s},\tilde{t}\right)\in\mathcal{P}_{xy}$
and $\l{P_{\tilde{s}v\tilde{t}}}\leq\frac{1}{\gamma}\l{P_{svt}}$
are at least $\left(\alpha\gamma\right)$-relative locally optimal.\label{lem:LO-acceptance}
\end{lem}
\begin{proof}
The $\T[\delta]$-test for $P_{svt}$ assured that $P_{svt}$ is $T$-locally
optimal with $T=\alpha\cdot\l{P_{svt}}$. Therefore, all paths $P_{\tilde{s}v\tilde{t}}$
with $\left(\tilde{s},\tilde{t}\right)\in\mathcal{P}_{xy}$ are also
$T$-locally optimal with $T=\alpha\cdot\l{P_{svt}}$. The local optimality
factor $\alpha_{\tilde{s}v\tilde{t}}$ of paths $P_{\tilde{s}v\tilde{t}}$
with $\left(\tilde{s},\tilde{t}\right)\in\mathcal{P}_{xy}$ and $\l{P_{\tilde{s}v\tilde{t}}}\leq\frac{1}{\gamma}\l{P_{svt}}$
is therefore at least
\begin{equation}
\alpha_{\tilde{s}v\tilde{t}}=\frac{T}{\l{P_{\tilde{s}v\tilde{t}}}}\geq\frac{T}{\frac{1}{\gamma}\l{P_{svt}}}=\frac{\gamma\alpha\l{P_{svt}}}{\l{P_{svt}}}=\alpha\gamma.
\end{equation}
That is, the paths $P_{\tilde{s}v\tilde{t}}$ are at least $\left(\alpha\gamma\right)$-relative
locally optimal.
\end{proof}
Following Lemma \ref{lem:LO-acceptance}, we can accept all pairs
$\left(\tilde{s},\tilde{t}\right)\in\mathcal{P}_{uw}$ with $\l{P_{\tilde{s}v\tilde{t}}}\leq\frac{1}{\gamma}\l{P_{svt}}$.
We do this in the same manner as we rejected paths. Let $\mathcal{A}_{xy}\subseteq\mathcal{P}_{xy}$
be defined as in the previous section. Since $\mathcal{P}_{xy}\backslash\mathcal{A}_{xy}$
contains only pairs $\left(\tilde{s},\tilde{t}\right)$ with $\l{P_{\tilde{s}v\tilde{t}}}<\l{P_{svt}}$,
which have been processed before $P_{svt}$, we only need to consider
the pairs in $\mathcal{A}_{xy}$ and classify all not yet processed
v-paths $P_{\tilde{s}v\tilde{t}}$ with $\left(\tilde{s},\tilde{t}\right)\in\mathcal{A}_{xy}$
and $\l{P_{\tilde{s}v\tilde{t}}}\leq\frac{1}{\gamma}\l{P_{svt}}$
as admissible. The described procedure to reject and accept multiple
paths at once is outlined in Algorithm \ref{alg:Multi-LO-test}.

\begin{algorithm}
\setstretch{1.35}
\SetKw{And}{and}
$R := \emptyset$\tcp*[r]{set of approximately admissible paths}
\ForEach{vertex $v \in \vvia$}{ 
	Let $\mathcal{P}$ be the set of all origin-destination combinations for which $v$ is a potential via vertex\;
	Sort the pairs in $\mathcal{P}$ in increasing order of the lengths of their v-paths\;
	\While{$\mathcal{P} \neq \emptyset$}{
		$(s,t):=$ next origin-destination pair in $\mathcal{P}$\;
		Do a $\T[\delta]$-test for the path $P_{svt}$ via $v$\;
		\eIf{the test fails and finds a suboptimal section $P_{uvw} \subseteq P_{svt}$}{
			\ForEach{pair $(s',t') \in \mathcal{P}$}{
				\If{$P_{uvw} \subseteq P_{s'vt'}$}{
					Remove $(s',t')$ from $\mathcal{P}$\;
				}
			}
		}{
			Add $P_{svt}$ to $R$\;
			Let $P_{xvy} \subseteq P_{svt}$ be the subsection of $P_{svt}$ that has been checked for local optimality\;
			\ForEach{pair $(s',t') \in \mathcal{P}$}{
				\If{$P_{xvy} \subseteq P_{s'et'}$ \And $\gamma \cdot \l{P_{s'vt'}} \leq \l{P_{svt}}$}{
					Add $P_{s'vt'}$ to $R$\;
					Remove $(s',t')$ from $\mathcal{P}$\;
				}
			}
		}
	}
}
\Return $R$\;

\caption[Testing whether the potentially admissible paths are approximately
$\alpha$-relative locally optimal]{Testing whether the potentially admissible paths are approximately
$\alpha$-relative locally optimal. \label{alg:Multi-LO-test}}
\end{algorithm}

\subsubsection{Optimization: using previous shortest path queries to determine locally
optimal subsections}

The outlined speedups become even more effective if the results of
individual shortest path queries are reused. Therefore, we save all
vertex pairs $\left(u,w\right)$ for which we know that $P_{uvw}=P_{uw}$.
 Note that we do not have to save unsuccessful shortest path tests,
because all v-paths $P_{\tilde{s}v\tilde{t}}$ with $P_{uvw}\subseteq P_{\tilde{s}v\tilde{t}}$
will be rejected right after $P_{uvw}$ has been found to be suboptimal
(see section \ref{subsec:LO-results-for-many-pairs}).

\sloppy The gain obtained from reusing shortest path results decreases
as the considered paths become longer. Since we are considering paths
in increasing order of lengths, the lengths of the subsections that
are required to be optimal increase as well. Therefore, the results
of earlier shortest path queries are of limited value if they are
only used as a lookup table.

However, we can exploit that due to the $\delta$-approximation, the
shortest path queries in the $\T[\delta]$-test typically consider
sections longer than required. The $\T[\delta]$-test conducts shortest
path queries between vertices $u$ and their partners $w:=\partner[t]u{\delta T}$.
Choosing $\delta>1$ not only reduces the number of necessary shortest path
queries but also makes the algorithm reject admissible paths. Therefore,
a test that sets $w:=\partner[t]u{\tau}$ for some $\tau\in[T,\delta T]$
will perform at least as well as the original algorithm.

With this observation, we can reuse previous shortest path results
as follows: when we search for the partner $w:=\partner[t]u{\delta T}$
of a vertex $u$, we test for all intermediate visited vertices $\tilde{w}:=\partner[t]u{\tau}$
with $\tau\leq\delta T$ whether the subpath $P_{uv\tilde{w}}$ is
known to be optimal. If such a vertex $\tilde{w}$ is found and $\tau\geq T$,
we accept $\tilde{w}$ as the partner of $u$ and progress as usual.

\subsection{Preprocessing}

Before REVC can be applied, a preprocessing step is required.  If
the set of origins and destinations of interest is known a priori,
we may start by reducing the graph by deleting dead ends that do not
lead to any of the considered origins and destinations.  In a second
step, we may add a random perturbation to the edge lengths to make
it easier to identify identical paths based on their length. As the
road costs (length, travel time, or other) are usually known with
limited precision, small perturbations will typically not change the
results significantly.

After these preparation steps, we can follow the preprocessing algorithm
by \citet{raman_reach_2006}. The algorithm determines upper bounds
on the reaches of vertices. To gain efficiency, the algorithm introduces shortcut
edges, which may bias the results so that admissible paths are falsely
rejected. However, it is easy to impose a length constraint on the
shortcut edges to reduce the introduced error. If REVC is applied
to a set of origins and destinations known in the preprocessing phase,
vertices bypassed by shortcut edges can be removed completely from
the graph. This increases the efficiency further.

The preprocessing step concludes with computing the shortest distances
between all origins and destinations. This can either be done with
individual shortest path queries for all origin-destination combinations
or in a single effort involving only one shortest path tree per origin-destination
pair. Either way, this step usually does not add significantly to
the algorithm's overall runtime. If the origins and destinations are
not known at the reprocessing time, this step can be postponed to
the execution of REVC.

\section{Tests}

To test the performance of REVC and to assess how input parameters
and the introduced optimizations affect the results and the computational
efficiency, we applied REVC to random route finding scenarios. To
gain insights into the algorithms' validity in modelling applications,
we tested how well the resulting paths are suited to predict observed
traveller behaviour. Below we provide details about our implementation
of REVC and the applied test procedures. Afterwards we present the
test results.

\subsection{Implementation}

We implemented REVC in the high-level programming language Python
(version 3.7) in combination with the numerical computing library
Numpy (version 1.16) and the software Cython (version 0.29), which
we used in particular to build a C extension for the shortest path
search. Despite our efforts to reduce bottle necks via C extensions,
a low-level implementation of REVC can be expected to be faster by
orders of magnitude. We computed shortest paths with the algorithm
RE \citep{raman_reach_2006}. The code used in this paper can be retrieved
as package ``lopaths'' from the Python Package Index (see \href{https://pypi.org/project/lopaths/}{pypi.org/project/lopaths}).
We executed our code in parallel on a Linux server with an Intel Xeon
E5-2689 CPU ($20$ cores with $3.1\,\text{GHz}$) and with $512\,\text{GB}$
RAM.

\subsection{Test methods}

\subsubsection{Test graph}

We tested REVC by applying it to a road network modelling the Canadian
province British Columbia (BC). The graph had $1.36$ million vertices
and $3.16$ million edges weighted by travel time. When we preprocessed
the graph, we limited the length of shortcut edges to $20\,\text{min}$,
which was less than $3\%$ of the mean shortest travel time between
the considered origins and destinations. For the empirical tests,
we joined the British Columbian road network with a graph representation
of the North American highway network. This additional network had
$2$ thousand vertices and $5.6$ thousand edges. 

\subsubsection{The effect of input parameters on the results and computation time}

We used a Monte Carlo approach to assess the effect of different input
parameters on the performance and the results of REVC. Specifically,
we considered the local optimality constant $\alpha$, the length
factor $\beta$, the approximation parameters $\gamma$ and $\delta$,
and the numbers of origins and destinations. We randomly generated
$10$ route finding scenarios ($20$ for tests on $\gamma$ and $\delta$)
and computed the mean and standard deviation of the results.

For each of these scenarios, we selected the origin and destination
locations randomly from the graph's vertices. We generated $10$ (+$10$
for tests on $\gamma$ and $\delta$) sets of origins and destinations,
which we reused for each assessed parameter combination to reduce
random influences on the results. When we varied the numbers of origins
and destinations, we increased the origin and destination sets as
necessary.

To measure the performance of the algorithm, we noted its total execution
time and the execution time per resulting path. Furthermore, we determined
the slowdown factor \citep[see][]{abraham_alternative_2013}, denoting
the ratio between the execution time of REVC and the corresponding
pair-wise shortest path search. In contrast to the execution time,
the slowdown factor is not strongly affected by the implementation
and hardware, since both REVC and the shortest path queries are run
with the same software on the same machine. Therefore, the slowdown
factor may be a more meaningful performance measure than the execution
time.

Note that it is possible to execute shortest path queries between
many origin-destination pairs in linear time of the origins and destinations
\citep{kliemann_route_2016}. However, the pair-wise approach used
to compute the slowdown factor provides a better comparison to pair-based
algorithms used in route choice modelling. Therefore, we applied the
pair-wise approach.

For a general assessment of the resulting paths, we determined the
average number and distribution of identified approximately admissible
paths and the mean length of these paths. These metrics may provide
hints on which parameter combinations are suitable in different modelling
applications. 

\subsubsection{Assessment of optimizations}

To assess the importance of the different optimizations we introduced
to make REVC computationally efficient, we executed the algorithm
repeatedly with different optimization steps disabled. We examined
the optimizations of (1) the growth bound for the shortest path trees,
(2) the tree pruning procedure, (3) the elimination of identical paths,
(4) the joint local optimality tests for multiple paths, and (5) reusing
shortest path query results. We applied the same randomized test procedure
as outlined in the previous section and executed the algorithm
with local optimality constant $\alpha=0.2$, length factor $\beta=1.5$,
and optimization constants $\gamma=0.9$ and $\delta=1.1$. We determined
the algorithm's execution time after disabling one optimization at
a time and computed the resulting relative changes in computation
time as compared to the fully optimized algorithm. To examine the
importance of the optimizations on different problem scales, we repeated
the tests with different numbers of origins and destinations. 

Disabling the respective considered optimizations was done as follows. 
(1) We examined the role of the optimized shortest path tree
growth bound by resetting the bound to the naive value $\beta\cdot\maxo{t\in\dests}\l{P_{st}}$
with origin set $O$, destination set $D$, and $\left(s,t\right)\in O\times D$.
(2) We tested the importance of the optimized pruning procedure in
two steps. First, we pruned only vertices $v$ with $\reach v<\alpha\cdot\cost v/2$
(see \citealp{abraham_alternative_2013}). Second, we used our stronger
pruning condition (equation (\ref{eq:reach-bound-1})) but refrained
from pruning even more vertices when growing backward shortest path
trees (equation (\ref{eq:reach-bound-2})). (3) We examined a simplified
algorithm to eliminate identical paths as well as skipping this step
completely. To simplify the algorithm, we skipped the step of eliminating
vertices that represent the same v-paths as their neighbours (section
\ref{subsec:Eliminating-similar-paths-neighbours}). (4) To test the
significance of joint local optimality tests of multiple paths, we
tested each route individually. (5) We examined the gain from reusing
shortest path query results by running the algorithm without saving
these results. 

\subsubsection{Empirical tests}

To test the empirical validity of the generated route choice sets,
we used data from road-side surveys in which travellers were surveyed
for their origins and destinations. Based on these data, we determined
which of the survey locations were passed frequently by travellers
driving between certain origins and destinations. By this means, we
obtained for each considered origin-destination pair a set of intermediate
locations where travellers were observed (called \emph{observed positive})
and a set of locations where travellers were not observed (called
\emph{observed negative}). If travellers choose admissible routes
as hypothesized, then the ``observed positive'' locations will be
on admissible routes for some reasonable parameters $\alpha$ and
$\beta$, and the ``observed negative'' locations will be on inadmissible
routes.

To see whether this was the case for our empirical observations, we
applied REVC to compute admissible routes between the considered origins
and destinations and classified all survey locations on admissible
routes as \emph{predicted positives}. The remaining survey locations
were considered \emph{predicted negatives}. Then, we determined (1)
the \emph{true positive rate}, i.e. the fraction of ``observed positive''
survey locations that were also ``predicted positive'', and (2)
the \emph{false positive rate}, the fraction of ``observed negative''
locations that were ``predicted positive''. We repeated this procedure
for different values of the local optimality constant $\alpha$ and
plotted the true positive rates against the false positive rates.
The resulting curve is the so-called receiver operating characteristic
(ROC), which is a widely used tool to assess the performance of classification
algorithms \citep{hosmer_applied_2013}. The area under the curve
(AUC) is a measure for the overall performance of the classifier \citep{hanley_meaning_1982}.
Large AUC values correspond to large true positive rates and small
false positive rates and thus indicate a good discrimination of positive
and negative observations. Since the set of admissible routes depends
not only on the local optimality constant $\alpha$ but also on the length
factor $\beta$, we computed the ROC and AUC for different values
of $\beta$. 

In addition to the ROC and AUC, we also determined how small local
optimality constant $\alpha$ must be chosen to cover $95\%$ or $100\%$ of the
positive observations. This result is of particular interest if the
admissible routes are filtered further before they are used in route
choice models. In this case, the false positive rate is of minor concern,
and the goal is to identify as many used routes as feasible. 

We based our analysis on survey data collected at watercraft inspection
stations in British Columbia in the years 2015 and 2016. These inspection
stations are set up to prevent human-mediated spread of aquatic invasive
species, and all road travellers transporting watercraft are required
to stop at these locations. We considered all inspection locations
where more than $50$ survey shifts were conducted in total. The mean
survey shift length at these $12$ locations was about $7$ hours.
Travellers were surveyed for the origin and destination waterbody
of their watercraft and nearby cities. To ensure that the traffic
was sufficiently dense to distinguish frequently used routes from
others, we considered origin-destination pairs for which more than
$50$ travellers were observed in total. This were $13$ pairs with
$5$ different origins and $7$ different destinations. The origins,
destinations, and survey locations are displayed in Figure \ref{fig:Origins-destinations-sampleLoc}.

To discern which survey locations were located on commonly used routes,
we used a threshold value for the mean number of observed travellers
per survey shift. Locations were classified as ``observed positive''
for an origin-destination pair if and only if the corresponding mean
traveller count exceeded the threshold value. Using a threshold value
has two advantages, namely (1) to reduce the potential bias resulting
from differing survey effort at different survey locations, and (2)
to reduce noise due to possible sampling error and travellers with
highly uncommon behaviour. To assess the impact of the threshold value
on the results, we considered different threshold values ranging between
a small positive value $\epsilon>0$ (any traveller observation results
in a positive classification) and $3$ observations per $100$ inspection
shifts. Depending on the threshold value, the number of ``observed
positive'' survey locations per origin-destination pair ranged between
$2.23$ and $4.15$, and the mean count of distinct origin-destination
pairs observed per survey location ranged between $2.42$ and $4.5$,
with one survey location not being passed by any traveller of interest.
\begin{figure}
\begin{centering}
\includegraphics{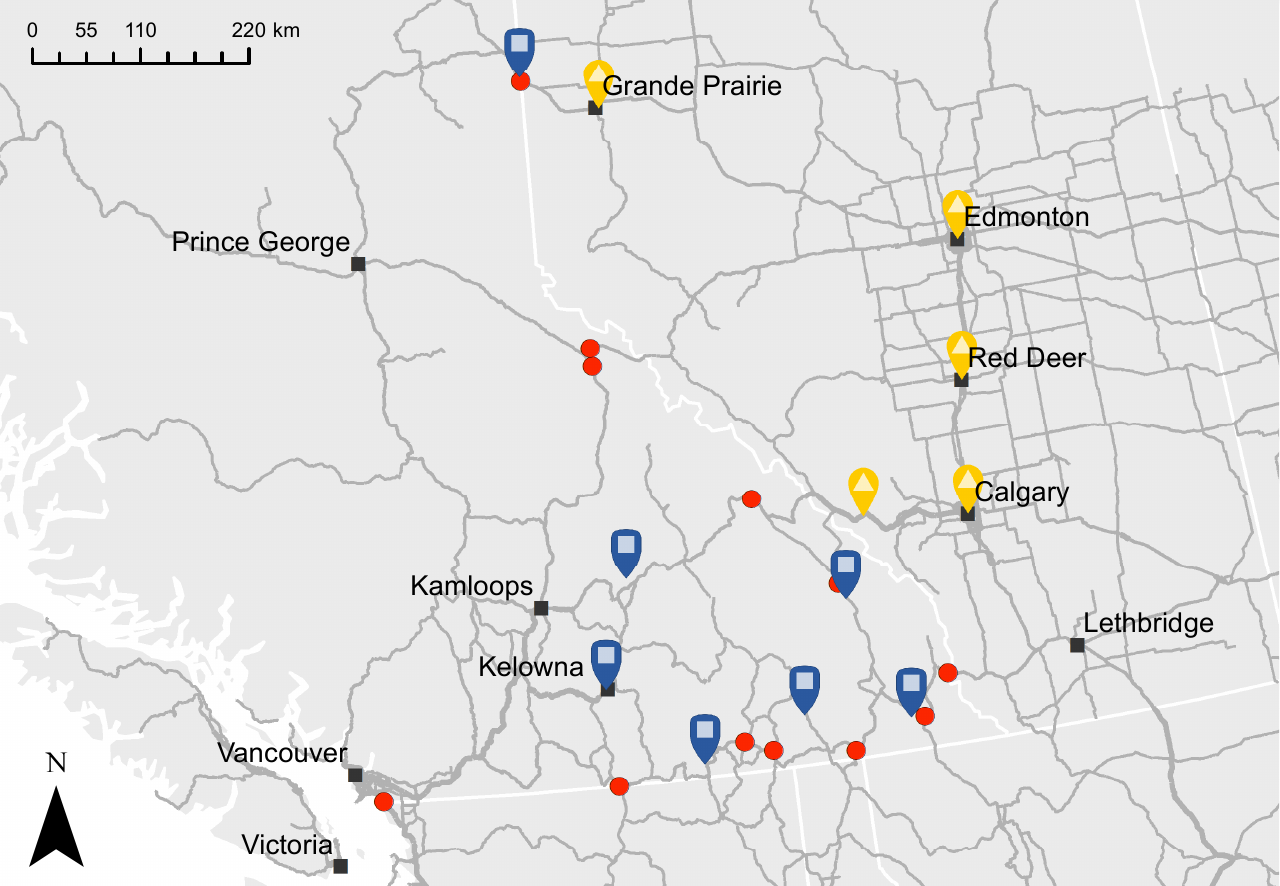}
\par\end{centering}
\caption{Considered origins, destinations, and sampling points. The origins
are shown as yellow markers with triangles and the destinations as
blue markers with squares. The traffic survey locations are depicted
as red circles. \label{fig:Origins-destinations-sampleLoc}}

\end{figure}

\subsection{Test results}

Below we provide the results of our tests. First, we focus
on the general results from the randomized experiments before we describe
the results of the tests involving empirical data.

\subsubsection{The effect of input parameters on the results and computation time}

\begin{figure}
\begin{tabular}[t]{>{\raggedleft}p{1cm}>{\centering}p{0.2\textwidth}>{\centering}p{0.2\textwidth}>{\centering}p{0.2\textwidth}>{\centering}p{0.2\textwidth}}
\begin{turn}{90}
\end{turn} & \raggedleft{}\textsf{\footnotesize{}\hspace*{0.8cm}}%
\begin{minipage}[t]{3cm}%
\begin{center}
\textsf{\footnotesize{}A}
\par\end{center}%
\end{minipage} & \raggedleft{}\textsf{\footnotesize{}\hspace*{0.8cm}}%
\begin{minipage}[t]{3cm}%
\begin{center}
\textsf{\footnotesize{}B}
\par\end{center}%
\end{minipage} & \raggedleft{}\textsf{\footnotesize{}\hspace*{0.8cm}}%
\begin{minipage}[t]{3cm}%
\begin{center}
\textsf{\footnotesize{}C}
\par\end{center}%
\end{minipage} & \raggedleft{}\textsf{\footnotesize{}\hspace*{0.8cm}}%
\begin{minipage}[t]{3cm}%
\begin{center}
\textsf{\footnotesize{}D}
\par\end{center}%
\end{minipage}\tabularnewline
\begin{turn}{90}
\textsf{\footnotesize{}\hspace*{0.4cm}}%
\begin{minipage}[t]{3cm}%
\begin{center}
\textsf{\footnotesize{}Execution time {[}$\text{s}${]}}
\par\end{center}%
\end{minipage}
\end{turn} & \textsf{\footnotesize{}\includegraphics[width=0.22\textwidth]{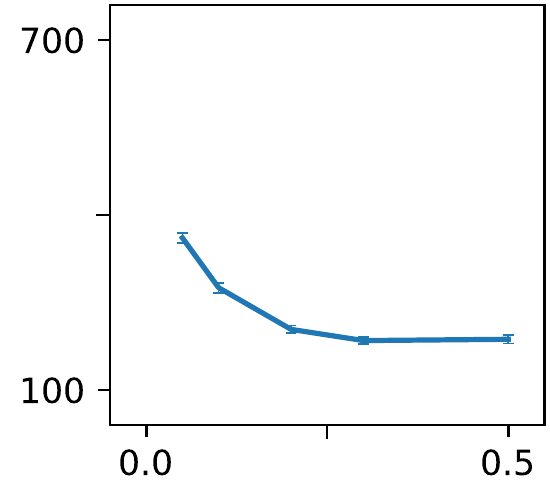}} & \textsf{\footnotesize{}\includegraphics[width=0.22\textwidth]{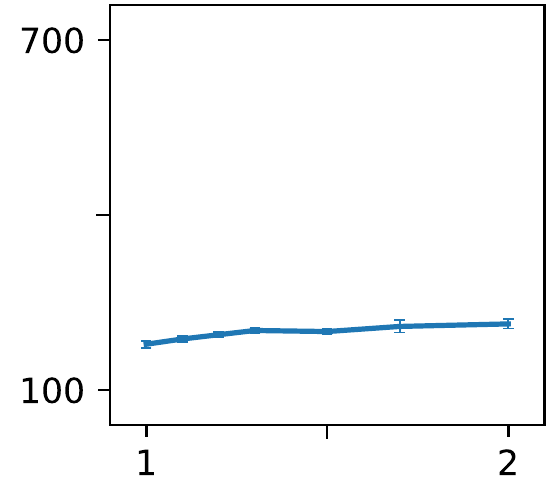}} & \textsf{\footnotesize{}\includegraphics[width=0.22\textwidth]{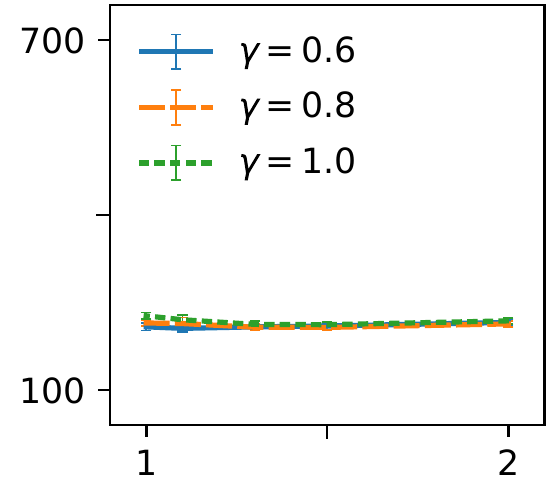}} & \textsf{\footnotesize{}\includegraphics[width=0.22\textwidth]{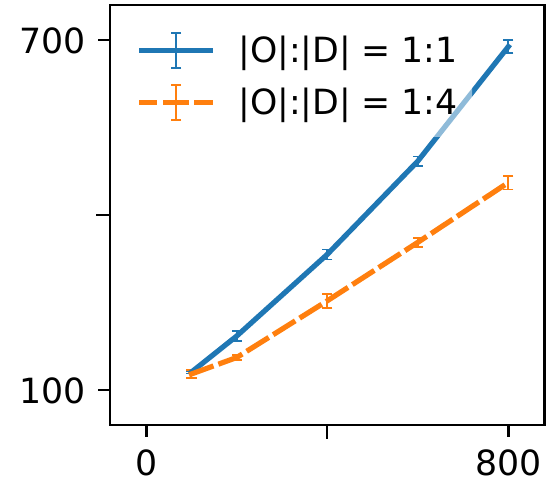}}\tabularnewline
\begin{turn}{90}
\noindent \raggedleft{}\textsf{\footnotesize{}\hspace*{0.4cm}}%
\begin{minipage}[t]{3cm}%
\begin{center}
\textsf{\footnotesize{}Time per path {[}$\text{ms}${]}}
\par\end{center}%
\end{minipage}
\end{turn} & \textsf{\footnotesize{}\includegraphics[width=0.22\textwidth]{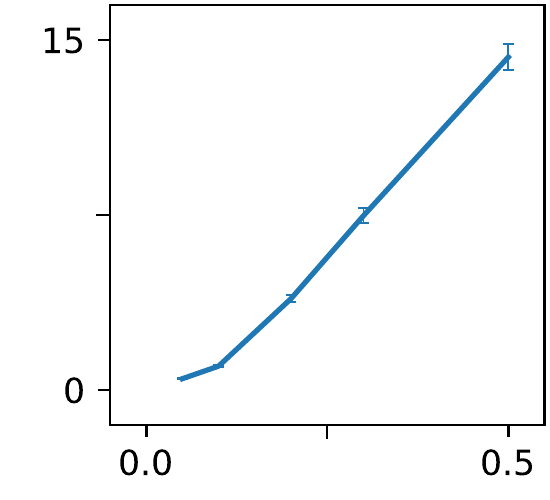}} & \textsf{\footnotesize{}\includegraphics[width=0.22\textwidth]{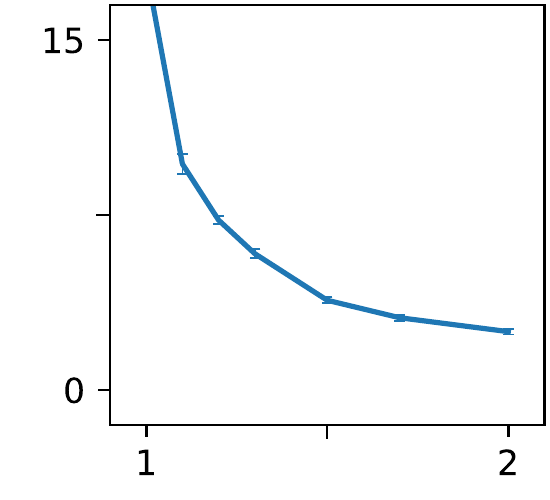}} & \textsf{\footnotesize{}\includegraphics[width=0.22\textwidth]{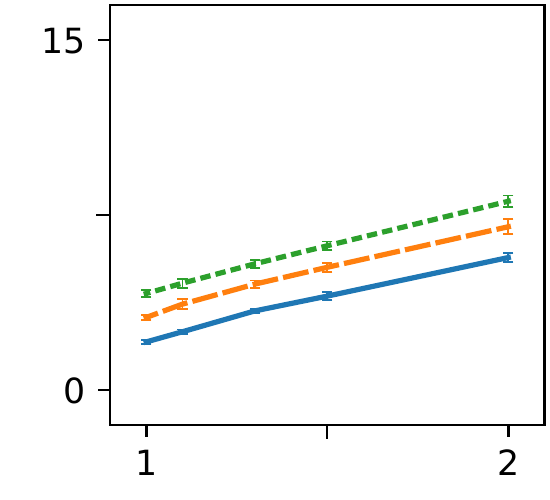}} & \textsf{\footnotesize{}\includegraphics[width=0.22\textwidth]{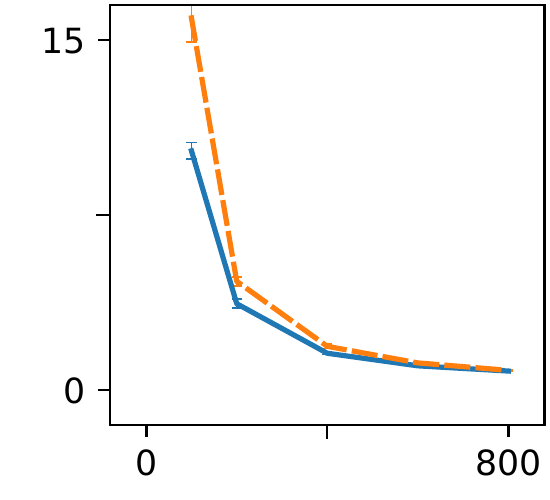}}\tabularnewline
\begin{turn}{90}
\noindent \raggedleft{}\textsf{\footnotesize{}\hspace*{0.4cm}}%
\begin{minipage}[t]{3cm}%
\begin{center}
\textsf{\footnotesize{}\vphantom {[} Slowdown factor}
\par\end{center}%
\end{minipage}
\end{turn} & \textsf{\footnotesize{}\includegraphics[width=0.22\textwidth]{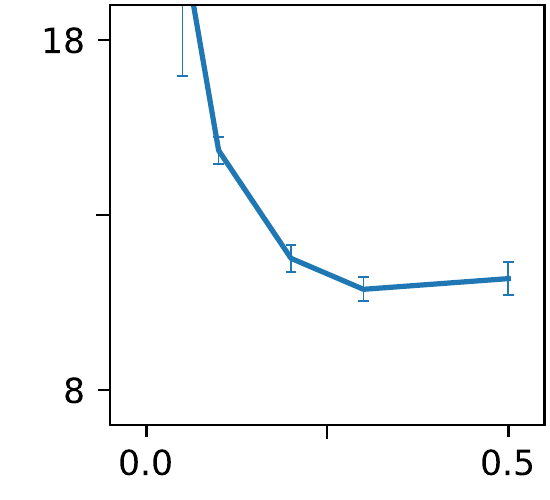}} & \textsf{\footnotesize{}\includegraphics[width=0.22\textwidth]{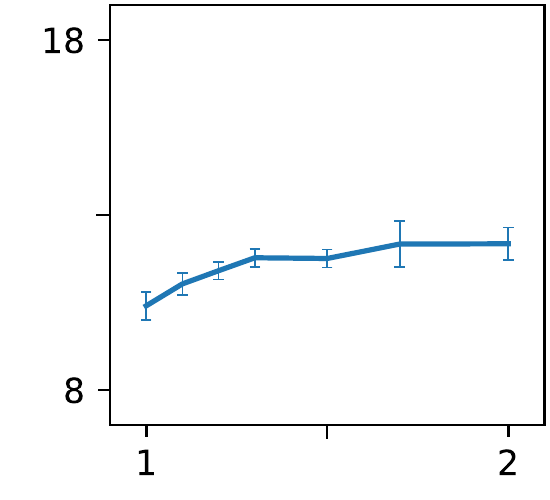}} & \textsf{\footnotesize{}\includegraphics[width=0.22\textwidth]{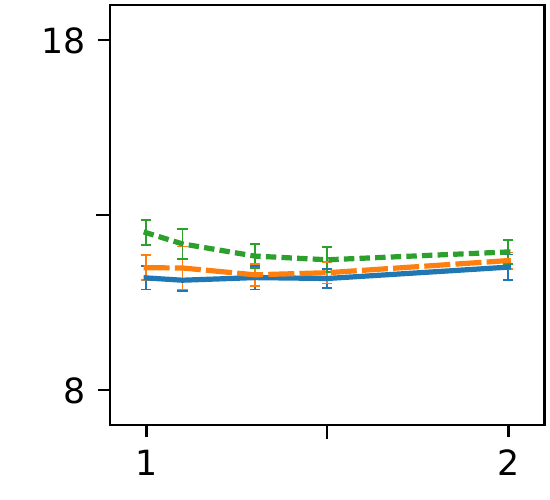}} & \textsf{\footnotesize{}\includegraphics[width=0.22\textwidth]{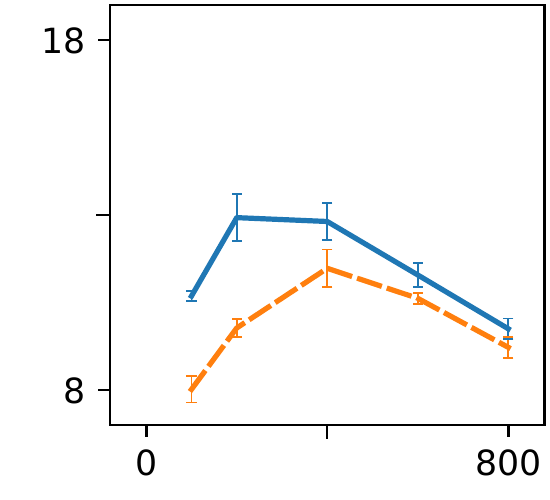}}\tabularnewline
\begin{turn}{90}
\textsf{\footnotesize{}\hspace*{0.4cm}}%
\begin{minipage}[t]{3cm}%
\begin{center}
\textsf{\footnotesize{}Resulting paths per orig.-dest. pair\vphantom {[} }
\par\end{center}%
\end{minipage}
\end{turn} & \textsf{\footnotesize{}\includegraphics[width=0.22\textwidth]{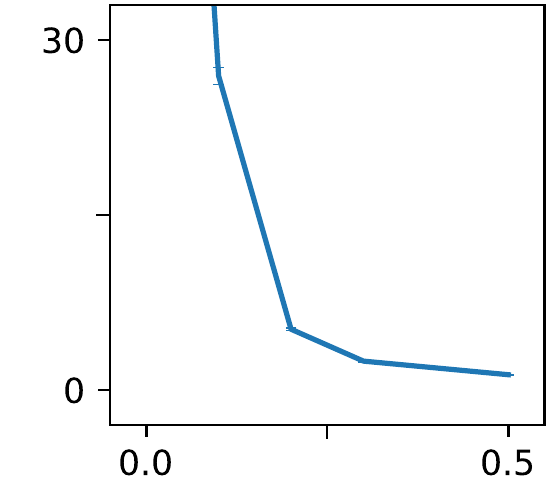}} & \textsf{\footnotesize{}\includegraphics[width=0.22\textwidth]{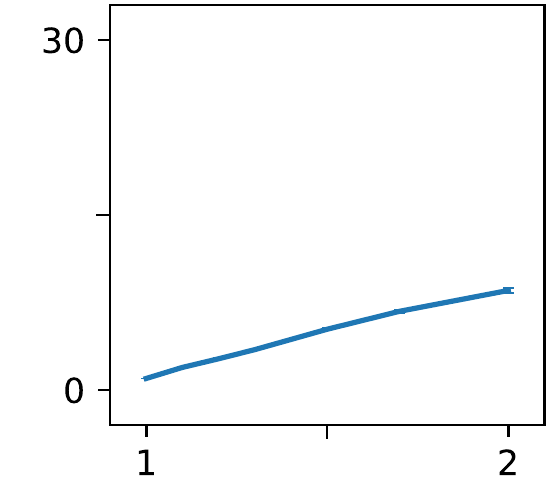}} & \textsf{\footnotesize{}\includegraphics[width=0.22\textwidth]{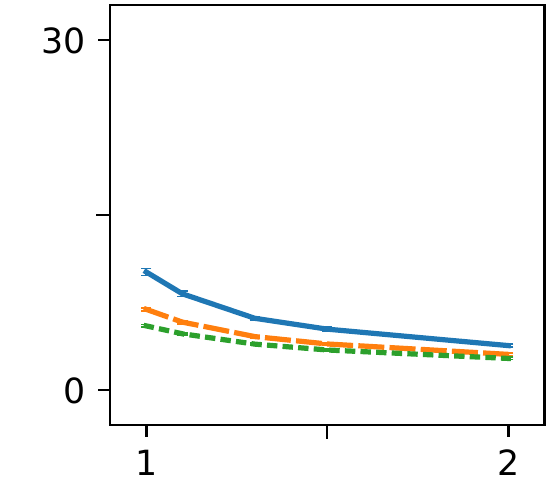}} & \textsf{\footnotesize{}\includegraphics[width=0.22\textwidth]{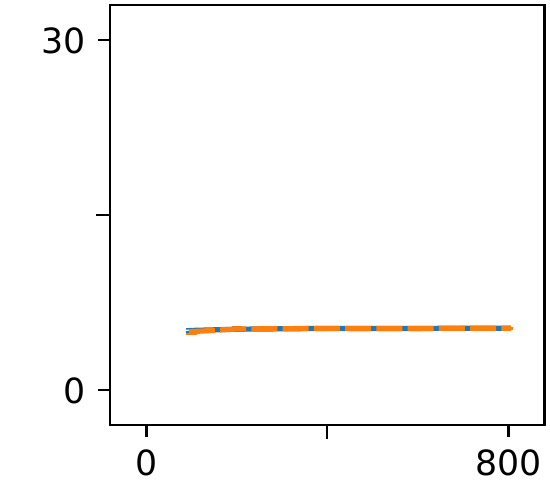}}\tabularnewline
\begin{turn}{90}
\textsf{\footnotesize{}\hspace*{0.4cm}}%
\begin{minipage}[t]{3cm}%
\begin{center}
\textsf{\footnotesize{}Mean path length {[}$\text{min}${]}}
\par\end{center}%
\end{minipage}
\end{turn} & \textsf{\footnotesize{}\includegraphics[width=0.22\textwidth]{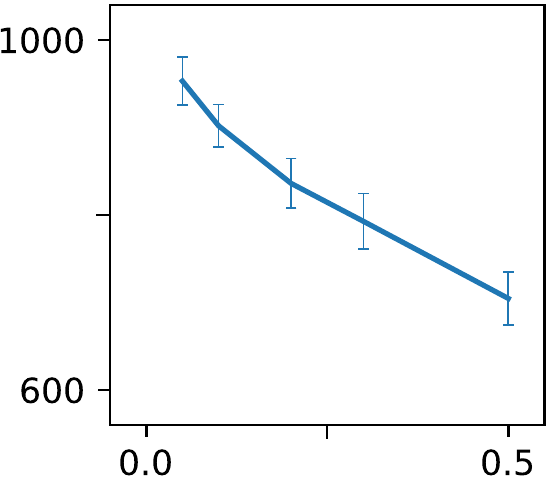}} & \textsf{\footnotesize{}\includegraphics[width=0.22\textwidth]{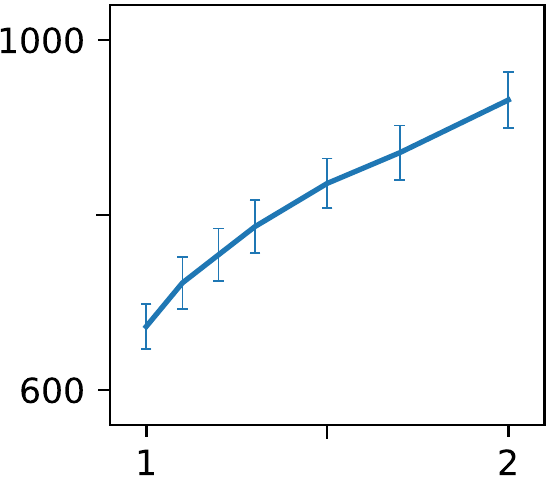}} & \textsf{\footnotesize{}\includegraphics[width=0.22\textwidth]{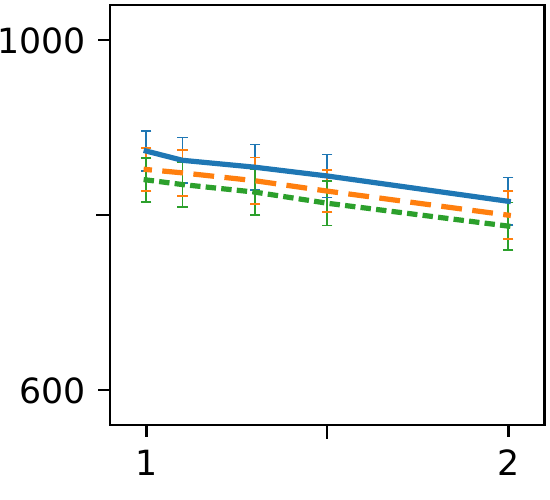}} & \textsf{\footnotesize{}\includegraphics[width=0.22\textwidth]{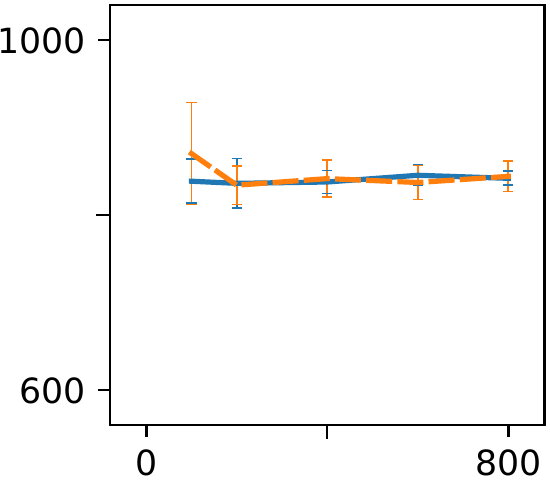}}\tabularnewline
\begin{turn}{90}
\end{turn} & \raggedleft{}\textsf{\footnotesize{}}%
\begin{minipage}[t]{2.6cm}%
\begin{center}
\textsf{\footnotesize{}Local optimality constant $\alpha$}
\par\end{center}%
\end{minipage} & \raggedleft{}\textsf{\footnotesize{}}%
\begin{minipage}[t]{2.5cm}%
\begin{center}
\textsf{\footnotesize{}Length factor $\beta$}
\par\end{center}%
\end{minipage} & \raggedleft{}\textsf{\footnotesize{}\hspace*{0.9cm}}%
\begin{minipage}[t]{3cm}%
\begin{center}
\textsf{\footnotesize{}Approximation constant $\delta$}
\par\end{center}%
\end{minipage} & \raggedleft{}\textsf{\footnotesize{}\hspace*{0.8cm}}%
\begin{minipage}[t]{3cm}%
\begin{center}
\textsf{\footnotesize{}Origins and desti- nations $\left|O\right|+\left|D\right|$}
\par\end{center}%
\end{minipage}\tabularnewline
\end{tabular}

\caption[Test results]{Test results. Different performance measures and result characteristics
are plotted against parameters. The whiskers depict the estimated
standard deviation. The line colours in column C correspond to different
values of the approximation constant $\gamma$. The line colours in
column D correspond to different ratios of origin number and destination
number. \protect \\
(Parameters unless specified otherwise: $\alpha=0.2$, $\beta=1.5$,
$\gamma=0.9$, $\delta=1.1$, $\left|O\right|=\left|D\right|=100$)
\label{fig:Test-results}}
\end{figure}

The results from the tests investigating the impact of the input parameters
on the algorithm's speed and results are displayed in Figure \ref{fig:Test-results}.
The constant $\alpha$, controlling the local optimality requirement,
had a strong influence both on the algorithm's running time and the
number of resulting paths. The effect of $\alpha$ on the execution
time levelled off at high values of $\alpha$. Decreasing $\alpha$
from $0.3$ to $0.05$ doubled the total execution time and reduced
the execution time per identified path by about factor $15$. In comparison,
increasing $\alpha$ from $0.3$ to $0.5$ had a minor effect only.
The mean number of paths followed a power law in $\alpha$ (exponent
$-1.84$). The length of the resulting paths decreased gradually as
$\alpha$ increased. An increase from $0.05$ to $0.5$ decreased
the mean length of admissible paths by about a quarter.

The parameter $\beta$, limiting the length of admissible paths, affected
the number and length of identified admissible paths but not the execution
time. The number of admissible paths increased almost linearly with
$\beta$; an increase of $0.1$ resulted in about $0.8$ additional
paths being found per origin-destination pair. Consequently, the execution
time per resulting path decreased with increasing $\beta$. The mean
lengths of the identified paths increased with their number. Raising
$\beta$ from $1$ to $2$ increased the mean path length by about
$40\%$.

The approximation parameters $\gamma$ and $\delta$ had little effect
on the execution time but a notable impact on the results. An increase
of $\gamma$ (increase in precision) consistently lengthened execution
times slightly. However, a decrease of $\delta$ (again, increase
in precision) \emph{reduced} the execution time per resulting path
and led to an optimal overall execution time at intermediate values
of $\delta$.

The number of identified paths varied more strongly than the execution
time when $\gamma$ and $\delta$ were changed. Dependent on the value
of $\delta$, decreasing $\gamma$ from $1$ to $0.6$ increased the
number of identified routes by $40\%$-$85\%$. Conversely, an increase
of $\delta$ from $1$ to $2$ decreased the number of identified
paths by more than $50\%$. The lengths of the resulting paths decreased
gradually both in $\gamma$ and $\delta$.

Changing the number of origins and destinations affected the execution
time but not the characteristics of the admissible paths. The execution
time increased almost linearly with the origin and destination number;
the slope depended on the origin to destination ratio. With
a ratio of $1:1$, the execution time increased by $80\,\text{s}$
per $100$ origins and destinations. With a ratio of $1:4$, the average
increase was $48\,\text{s}$ per $100$ origins and destinations.
The time per identified path and the slowdown factor decreased as
more origin and destination locations were added.

\begin{figure}
\captionsetup[subfigure]{position=top,justification=raggedleft,margin=0pt, singlelinecheck=off} 
\begin{centering}
\textsf{}\subfloat[]{\textsf{\includegraphics[scale=0.7]{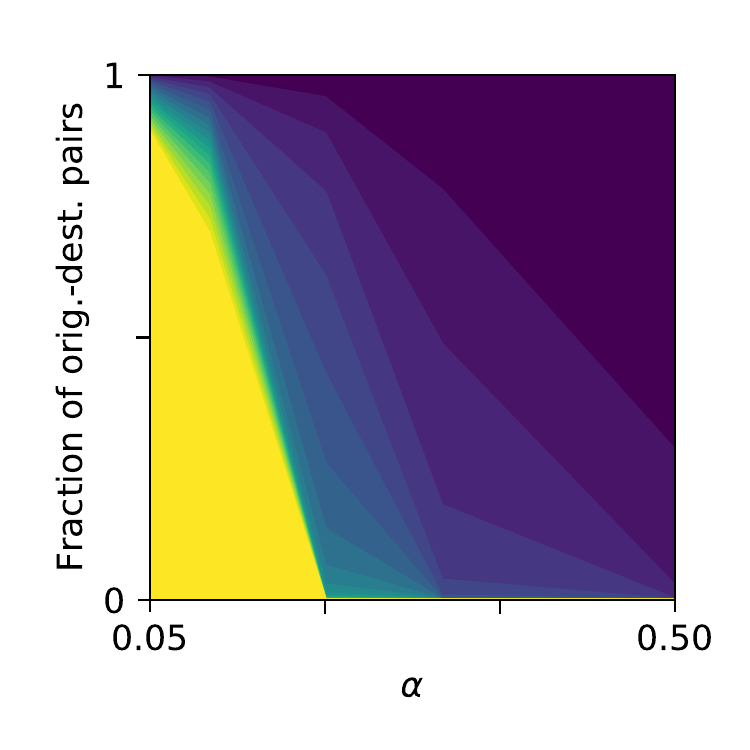}}}\textsf{\hspace*{2cm}}\subfloat[]{\textsf{\includegraphics[scale=0.7]{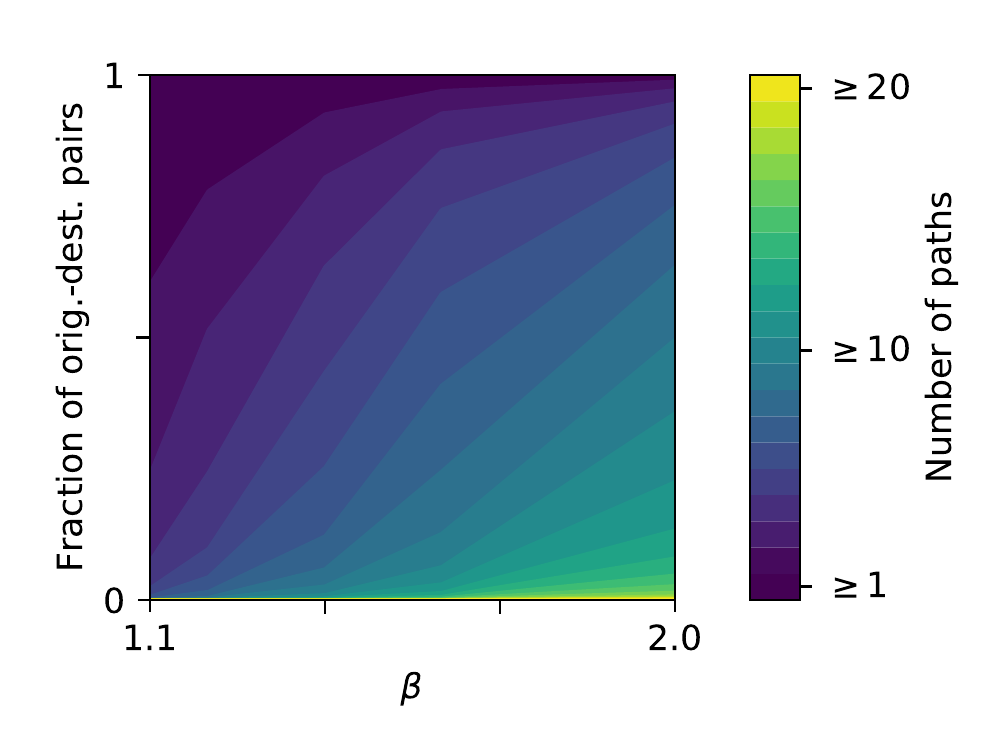}}}
\par\end{centering}
\caption[Distribution of paths dependent on the local optimality constant and
the length constant]{Distribution of paths dependent on (a) the local optimality constant
$\alpha$ and (b) the length constant $\beta$. The $y$-axis shows
which fraction of origin-destination pairs were connected by at least
the number of paths given by the colour.\label{fig:Distribution-of-paths}
The parameters are the same as in Figure \ref{fig:Test-results} column
A and B.}
\end{figure}

Figure \ref{fig:Distribution-of-paths} displays the distribution
of paths per origin-destination pair dependent on the local optimality
constant $\alpha$ and the length constant $\beta$. Many origin-destination
pairs are connected by numerous admissible paths if $\alpha$ is smaller
than $0.2$. For example, with $\alpha=0.1$ and $\beta=1.5$, about
three quarters of the origin-destination pairs were connected by more
than $20$ routes. In contrast, with $\alpha=0.3$, less than $0.7\%$
of the pairs were connected by more than $5$ paths, and $22\%$ of
the pairs were connected by the shortest path only. The latter fraction
increased to $72\%$ for $\alpha=0.5$.

The distribution of paths per origin-destination pair changed more
gradually with $\beta$. With $\alpha=0.2$, a large value of $\beta=2$
resulted in $99\%$ of the pairs being connected by multiple admissible
paths; $22\%$ were connected by more than $10$ paths. On
the other end of the spectrum, with $\beta=0.1$, $40\%$ of the origin-destination
pairs were connected by $1$ admissible path only and $0.6\%$ were
connected by more than $5$ admissible paths.

\subsubsection{Assessment of optimizations}

Assessing the role of the different optimizations we introduced to
make REVC more efficient, we obtained a broad spectrum of results,
shown in Table \ref{tab:Opt-Test-Results}. 

First, changing the growth bound for the shortest path trees had only
a small effect on the computation time. The changes in computation
time were smaller than the corresponding standard deviations. 

Second, growing the shortest path trees with less strict pruning procedures
increased the computation time by $11\%$-$56\%$. This the change
was particularly high when the number of origins and destinations
was imbalanced. There was only a small difference between disabling
all pruning optimizations and solely refraining from earlier pruning
in backward direction. 

Third, omitting the step of identifying vertices that represent the
same paths as their neighbours \emph{decreased} the computation time
by $2\%$-$8\%$. This effect was smaller the less balanced the numbers
of origins and destinations were. In contrast, refraining from any
identification of identical paths increased the computation time by
more than $70\%$ with a larger effect in scenarios with strongly
differing origin and destination numbers. 

Fourth, disabling the joint tests for local optimality led to large
changes in computation time (increase by factor $3$ up to factor
$11.8$). The increase was stronger the more origin-destination pairs
were considered. 

Lastly, stopping to reuse shortest path query results increased the
computation time moderately by $3\%$-$16\%$. The effect was strongest
when the number of origins was small and the number of destinations
large.

\begin{sidewaystable}
\begin{centering}
\footnotesize\renewcommand{\arraystretch}{1.7} \setlength\arrayrulewidth{0.5pt}%
\begin{tabular}{>{\raggedright}p{5.8cm}>{\centering}p{1.6cm}>{\centering}p{1.6cm}>{\centering}p{1.9cm}>{\centering}p{1.8cm}>{\centering}p{1.9cm}>{\centering}p{1.9cm}>{\centering}p{2cm}>{\centering}p{1.9cm}}
\toprule 
Disabled optimization & \multicolumn{4}{c}{Execution time {[}s{]} (standard deviation)} & \multicolumn{4}{c}{\% increase (standard deviation)}\tabularnewline
\raggedleft{}$\left|O\right|\times\left|D\right|$ & $50\times100$ & $100\times100$ & $200\times500$ & $50\times1000$ & $50\times100$ & $100\times100$ & $200\times500$ & $50\times1000$\tabularnewline
\midrule
\midrule 
None & $133$ ($4.1$) & $191$ ($6.6$) & $476$ ($15.0$) & $329$ ($11.5$) & -- & -- & -- & --\tabularnewline
\midrule 
Optimized shortest path tree height & $135$ ($4.0$) & $193$ ($5.8$) & $467$ ($10.6$) & $323$ ($11.4$) & $1.8$ ($4.3$) & $0.9$ ($4.6$) & $-1.9$ ($3.8$) & $1.8$  $(5.1$)\tabularnewline
\midrule 
Earlier pruning in backward direction & $158$ ($4.3$) & $213$ ($4.3$) & $558$ ($10.4$) & $493$ ($11.3$) & $18.5$ ($4.8$) & $11.4$ ($4.5$) & $17.3$ ($4.3$) & $52.5$  ($6.4$)\tabularnewline
\midrule 
All pruning optimizations & $158$ ($4.5$) & $216$ ($5.1$) & $556$ ($11.8$) & $505$ ($14.6$) & $18.8$ ($4.9$) & $13.3$ ($4.7$) & $16.9$ ($4.5$) & $56.3$  ($7.1$)\tabularnewline
\midrule 
Identifying neighbouring vertices representing identical paths & $123$ ($4.2$) & $177$ ($4.2$) & $452$ ($11.7$) & $317$ ($11.7$) & $-7.6$ ($4.2$) & $-7.1$ ($3.9$) & $-5.0$ ($3.9$) & $-2.1$ ($5.0$)\tabularnewline
\midrule 
Identifying identical paths & $237$ ($14.5$) & $331$ ($13.4$) & $887$ ($31.7$) & $683$ ($42.9$) & $77.9$ ($12.2$) & $73.4$ ($9.2$) & $86.4$ ($8.9$) & $111.5$  ($15.2$)\tabularnewline
\midrule 
Joint tests for local optimality & $400$ ($9.8$) & $735$ ($17.3$) & $5615$ ($105.3$) & $2849$ ($66.3$) & $200.2$ ($11.7$) & $284.8$ ($16.1$) & $1079.9$ ($43.4$) & $781.6$  ($37.3$)\tabularnewline
\midrule 
Reusing shortest path query results & $140$ ($4.9$) & $200$ ($9.3$) & $492$ ($9.9$) & $374$ ($13.0$) & $5.5$ ($4.9$) & $4.5$ ($6.1$) & $3.4$ ($3.9$) & $15.7$  ($5.7$)\tabularnewline
\bottomrule
\end{tabular}
\par\end{centering}
\caption{\label{tab:Opt-Test-Results}The impact that different optimizations
introduced with REVC have on the the algorithm's running time. For
each given optimization, the table displays (1) the running time of
REVC if this optimization were disabled and (2) the corresponding
relative increase in running time. The results are given for four
scenarios with different numbers of origins and destinations. The
standard deviations of the results are displayed in parenthesis. }
\end{sidewaystable}

\subsubsection{Empirical tests}

The ROC curves that we obtained for different length factors $\beta$
and classification thresholds are displayed in Figure \ref{fig:ROC}.
Quantitative results are given in Table \ref{tab:AUC-results}. For
large admissible length factors ($\beta\geq2$), the area under the
curve (AUC) constantly exceeded $0.9$. For smaller length factors
($1.3\leq\beta<2$), the AUC values were smaller but never below $0.78$. 

A moderate local optimality requirement of $\alpha=0.25$ sufficed
to cover $95\%$ of the positive observations regardless of how many
traveller observations were required for positive classification of
survey locations. In the scenario in which any observation sufficed
for positive classification of survey locations, one triplet of origin,
intermediate destination, and final destination was not covered by
any admissible path for the tested parameter values. In the remaining
scenarios, all ``observed positive'' locations were on admissible
routes for some $\alpha$ value. For a classification threshold of
$1$ observation per $100$ survey shifts, $\alpha$ had to be chosen
as low as $0.07$ to cover all positive observations. When the classification
threshold was large ($\geq2$ observations per $100$ survey shifts),
all ``observed positive'' locations were on $0.4$-relative locally
optimal paths. That is, these paths were optimal on all subsections
shorter than $40\%$ of the entire path. 
\begin{figure}
\captionsetup[subfigure]{position=top,justification=raggedleft,margin=0pt, singlelinecheck=off} 
\begin{centering}
\textsf{}\subfloat[]{\textsf{\includegraphics[scale=0.7]{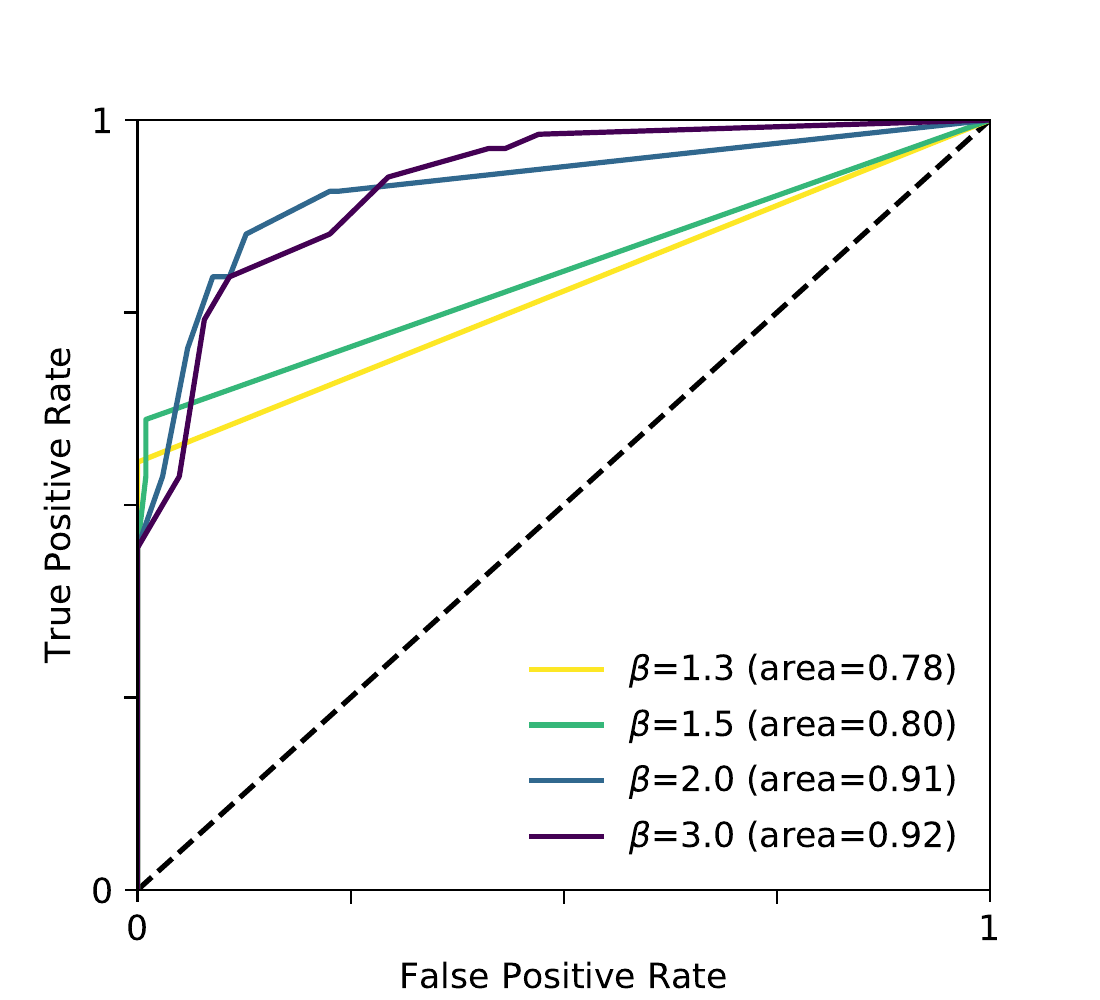}}}\textsf{\hspace*{2cm}}\subfloat[]{\textsf{\includegraphics[scale=0.7]{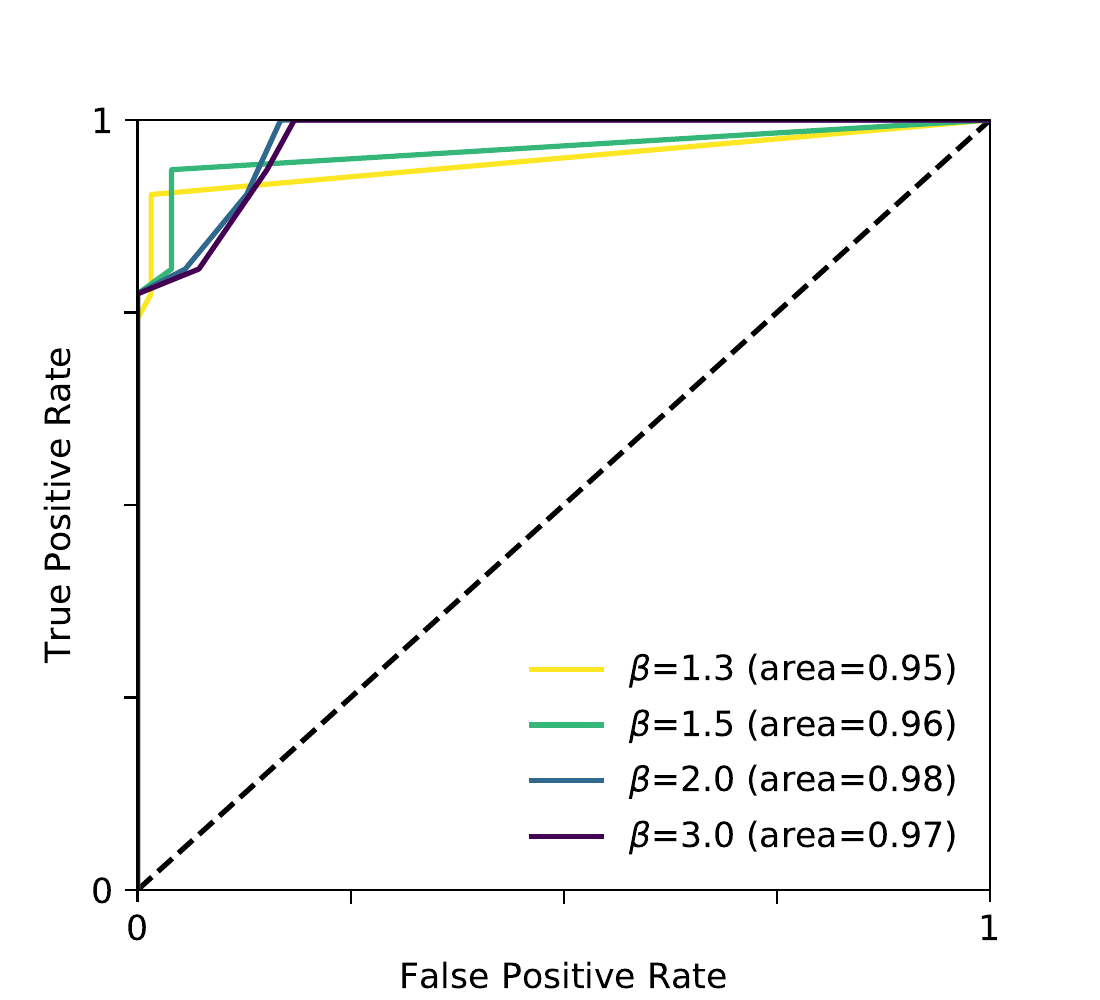}}}
\par\end{centering}
\caption{Receiver operating characteristic (ROC) curves for different length
factors $\beta$. In Subfigure (a), any location with observed travellers
was classified as ``observed positive'', whereas in Subfigure (b),
only locations with $2$ traveller observations per $100$ shifts
were classified ``observed positive''. In both scenarios, the locations
on admissible routes coincided strongly with the locations with positive
observations. This can be seen from the high true positive rates achieved
at the same time as small false positive rates. The dashed line shows
the performance of a hypothetical random classifier. A point with
true positive and false positive rate of $1$ was added to complete
the curve though these values did not occur in practice.\label{fig:ROC}}
\end{figure}

\begin{table}
\begin{centering}
\footnotesize\renewcommand{\arraystretch}{1.7} \setlength\arrayrulewidth{0.5pt}%
\begin{tabular}{>{\centering}p{3.2cm}>{\centering}p{1.5cm}>{\centering}p{1.4cm}>{\centering}p{3cm}>{\centering}p{3cm}>{\centering}p{3cm}}
\toprule 
Observed travellers per $100$ shifts required for positive classification & AUC 

($\beta=1.5$) & AUC 

($\beta=3$) & $\alpha$ required to cover $95\%$ of observed positives ($\beta=3$) & $\alpha$ required to cover all observed positives ($\beta=3$) & Fraction of ``observed positive'' locations on shortest routes\tabularnewline
\midrule
\midrule 
Any & $0.80$ & $0.92$ & $0.25$ & -- & $0.35$\tabularnewline
\midrule 
$1$ & $0.84$ & $0.94$ & $0.3$ & $0.07$ & $0.41$\tabularnewline
\midrule 
$2$ & $0.96$ & $0.97$ & $0.4$ & $0.4$ & $0.61$\tabularnewline
\midrule 
$3$ & $0.97$ & $0.97$ & $0.4$ & $0.4$ & $0.66$\tabularnewline
\bottomrule
\end{tabular}
\par\end{centering}
\caption{\label{tab:AUC-results}Classification results for different classification
thresholds. The AUC values are generally high and increase as more
traveller observations are required to classify a survey location
as ``observed positive''. The first column shows the traveller counts
per $100$ survey shifts required to classify a location as ``observed
positive''. The second and third column display AUC values obtained
with different path length constraints. The fourth and fifth column
contain the maximal value of the local optimality constant $\alpha$
for which $95\%$ or $100\%$ of the ``observed positive'' locations
were on admissible routes, respectively. The right-most column indicates
how many locations classified as ``observed positive'' were located
on the shortest routes between the respective origins and destinations. }
\end{table}

\section{Discussion}

We have introduced an algorithm that efficiently identifies locally
optimal paths between many origin-destination pairs and tested both
the algorithm's computational performance and its ability to predict
empirical traffic observations. Our algorithm REVC identifies all
approximately admissible routes between the origins and destinations,
and its execution time is driven by the number of distinct origins
and destinations rather than the number of origin-destination \emph{pairs}.
The empirical tests suggest that locally optimal routes are a suitable
tool to predict where travellers between specific origins and destination
are likely to be observed. These results combined indicate that REVC
is applicable in large-scale traffic models.

Our tests examining the impact of different input parameters on the
results and computation time show that REVC's performance depends
mostly on the local optimality constant $\alpha$ and the number of
origins and destinations. While the total execution time increases
with the number of considered origins and destinations and with reduced
$\alpha$, the execution time per identified path decreases. That
is, REVC becomes more efficient compared to repeated path queries
the more paths are generated.

The length bound $\beta$ had only a minor effect on the execution
time. This may be surprising, as an increase in $\beta$ allows more
vertices to be included in the shortest path trees. However, the impact
of $\beta$ is reduced by our pruning technique, which is most effective
for long paths. Furthermore, large parts of the graph had been scanned
for small values of $\beta$ already, since the considered origins
and destinations were distributed over the entire graph. Therefore,
few additional vertices were considered with increased $\beta$.

The effect of $\beta$ may be larger if all origin and destination
locations are located within a small subsection of the graph. Nonetheless,
in many modelling applications, the origin and destination locations
will be distributed over the whole considered road network. For example,
when the traffic from the outskirts of a city to downtown is modelled,
it is unlikely that travellers leave the greater metropolitan area.
Therefore, it is reasonable to consider an accordingly constrained
graph.

REVC applies approximations to gain efficiency. However, the approximation
constants had relatively small effects on the algorithm's performance
in our tests. This suggests that approximations may not always be
necessary. However, the benefit of the approximations will become
larger if the origin and/or destination vertices are not randomly
spread over the whole graph but located in constrained areas. Then,
partial results can be reused more effectively. As the admissibility
checks were responsible for a limited portion of the overall execution
time only, the gain of the approximations will also become more significant
if more paths have to be checked for local optimality.

An interesting observation is that intermediate values of the approximation
constant $\delta$ led to lower execution times than large values.
This is surprising, because smaller values of $\delta$ increase the
number of shortest path queries required in the $\T[\delta]$-test.
However, small values of $\delta$ have the advantage that the subsections
checked for local optimality get shorter. This makes it more likely
that test results can be reused to reject many inadmissible paths
at once. In point to point queries, the $\T[2]$-test (used by \citealp{abraham_alternative_2013})
may still be superior.

The tests evaluating the importance of the different optimizations
introduced with REVC resulted in a heterogeneous picture. The most
important innovation of REVC was the joint local optimality test of
many paths. This result was expected, since separate tests must consider
each origin-destination pair individually, thus making the algorithm's
runtime strongly dependent on the number of origin-destination pairs. 

Another significant speedup was obtained by rejecting identical paths
prior to local optimality checks. However, identifying and neglecting
vertices that represent the same v-paths as their neighbours decreased
the algorithm's efficiency despite having a positive effect on the
asymptotic runtime. This was due to our implementation of REVC, where
the computation time required to identify paths with identical lengths
was dominated by the number of origin-destination pairs rather than
the number of paths. Though omitting the first path comparison step
can apparently speed up the algorithm, the procedure can prove useful
if the origins and destinations are spatially separated, which allows
more vertices to be rejected in this step. 

A moderate speedup was gained by improving the pruning procedure applied
during the shortest path tree growth. Here, earlier pruning in
backward direction turned out to be the most important optimization.
Disabling this improvement only had almost the same effect as disabling
all pruning optimizations. This is because early pruning reduces shortest
path trees by a complete layer of leaf vertices that may need to be
considered in computationally expensive $\T[\delta]$-tests otherwise.

Reusing the results of shortest path queries had a small but notable
effect on computational efficiency. The efficiency gain is highest
if many v-paths via a vertex share subsections. This happens if origins
and destinations are spatially separated or if the numbers of origins
and destinations are imbalanced. Note, however, that reusing shortest
path query results increases the number of optimality checks and thus
facilitates the accuracy of the results.

The optimization of the shortest path tree growth bound had a minor
effect only. This result is in line with the small impact that the
length factor $\beta$ had on the computation time, and the explanation
for the result is similar. Consequently, the optimized tree growth
bound will become more important if all origins and destinations are
located in a small part of the considered graph. 

Besides assessing the computational performance of REVC, we also tested
the empirical validity of the computed routes. The tests showed that
the paths returned by REVC allow precise predictions of where individuals
travelling between given origins and destinations can be observed.
Typically, predictors with AUC values exceeding $0.8$ are considered
excellent and those with AUC values exceeding $0.9$ outstanding \citep{hosmer_applied_2013}.
The large AUC values we obtained, consistently greater than $0.9$
for $\beta=3$, suggest that local optimality can be a helpful criterion
to discriminate used roads from unused roads -- and thus to characterize
route choice sets. Though no tracking data were available to us that
would have allowed us to assess the overlap between observed and computed
routes, our survey points and data were sufficiently heterogeneous
to give significant insight into the validity of locally optimal routes
in modelling applications. 

\subsection{Significance}

Determining multiple paths between an origin and a destination based
on a local optimality criterion is a well established approach in
route planning research \citep{abraham_alternative_2013,delling_customizable_2015,luxen_candidate_2015,kliemann_route_2016}.
An obstacle hindering the application of these algorithms in route
choice models was that these algorithms return only few heuristically
chosen paths rather than the complete set of admissible paths. Furthermore,
these algorithms are based on an inflexible approximation whose impact
on the result was not exactly known. Our algorithm REVC solves these
issues. Though REVC may not be competitive in point to point queries,
the algorithm efficiently exploits redundancies occurring when many
origin-destination pairs are considered.

Generating route choice sets based on local optimality has multiple
advantages. The underlying principle is simple and has a sound mechanistic
justification. The optimality principle is applied on a local scale,
whereas the mechanisms governing travellers' overall route choices
do not need to be known. Therefore, no extensive data sets are needed
to generate choice sets. In addition, our empirical test results suggest
that local optimality is indeed a suitable criterion to distinguish
used roads from unused roads, yielding a high coverage of actual observations
and a low rate of false positive predictions.

Fitting the choice set parameters to data is a discrete optimization
problem and can therefore be challenging. REVC permits two free variables:
the local optimality parameter $\alpha$ and the length parameter
$\beta$. As the latter does not have a strong impact on the execution
time, $\beta$ can be chosen liberally, leaving $\alpha$ as the only
remaining free parameter. Optimizing $\alpha$, in turn, is comparatively
easy, as this is a one-dimensional problem.

Choice sets consisting of locally optimal v-paths are typically relatively
small while still covering a broad spectrum of different routes \citep[see ][]{abraham_alternative_2013}.
This agrees with our empirical tests, where a low rate of false positive
predictions was achieved at the same time as a high true positive rate.
The high specificity of local optimal routes allows for sophisticated
models for the second route choice step, in which travellers select
routes from the choice sets. The option to use sophisticated metrics
to measure the quality of the route candidates may improve the overall
model fit. In addition, using small choice sets also reduces a bias
observed in route choice models when many insignificant routes are
present in choice sets \citep{bliemer_impact_2008}. 

The favourable quality to quantity ratio of locally optimal v-paths
and the practically linear relationship between execution time and
origin and destination numbers make REVC particularly useful in comprehensive
traffic models. In such applications, many origin-destination pairs
have to be considered, and the computed choice sets need to be kept
in memory for further processing. This makes it difficult to apply
methods based on point to point queries, such as link elimination
\citep{azevedo_algorithm_1993}, link penalty \citep{de_la_barra_multi-dimensional_1993},
or constrained enumeration methods \citep{prato_applying_2006}. Similar
challenges face algorithms that need to generate many paths, such
as stochastic approaches or methods that include a filtering step
to select admissible paths from a large number of candidates \citep[see][]{bovy_modelling_2009}.
Therefore, REVC may be of specific use in comprehensive models.

The results of REVC provide insights into the distribution and properties
of locally optimal routes in real road networks. In our tests, the
number of admissible paths decreased with $\alpha$ in a power law
relationship, whereas it increased linearly in $\beta$. Such experimental
results could be the starting point for a more in-depth theoretical
analysis of the distribution of locally optimal routes in road networks.
The resulting insights may facilitate the development of new algorithms.

The experimental results are also valuable as benchmarks for existing
algorithms searching locally optimal v-paths for route planning purposes
\citep{abraham_alternative_2013,hutchison_alternative_2013,luxen_candidate_2015}.
Some of these algorithms apply approximations to gain efficiency.
The presented results can help to assess the impact of these approximations.
Our results suggest that the applied $\T[2]$-approximation falsely
rejects half of the admissible paths.

In addition to assessing the accuracy of faster algorithms, the complete
sets of admissible paths generated with REVC can also be used to evaluate
the success rate and the quality of the paths generated with these
algorithms. Note, however, that our definition of admissible paths
deviates slightly from the definition applied in earlier papers. Refer
to Appendix \ref{sec:APX-Comparison-of-REV-REVC} for details.

REVC contains several optimizations that can be directly applied to
make the family of algorithms based on REV more efficient. These optimizations
include the improved bounds for tree growth and pruning as well as
the idea to exclude u-turn paths by considering via edges. Similarly,
the $\T[\delta]$-test can be directly applied to increase the accuracy
of all algorithms using the $\T$-test. Our randomized tests can be
used to assess the benefit gained from the different optimizations.
Hence, this paper may also contribute to make route planning software
more efficient. We provide a more in-depth discussion in Appendix
\ref{sec:APX-Comparison-of-REV-REVC}.

\subsection{Limitations}

REVC focuses on single-via paths. A complete search for locally optimal
routes should not limit the set of considered paths. However, considering
v-paths can be justified by assuming that travellers may drive via
an intermediate destination. Furthermore, the focus on v-paths excludes
zig-zag routes, which may be deemed unrealistic. Therefore, a criterion
limiting the set of admissible paths may not only be a computational
necessity but also beneficial in route choice models.

Nonetheless, REVC may be extendable to include paths via two intermediate
destinations. Road networks usually have a small set $W$ of vertices
so that every sufficiently long shortest path includes at least one
of these vertices \citep{abraham_highway_2010}. If $W$ could be
identified efficiently, REVC could be applied to compute v-paths from
the origins to the vertices in $W$ and from the vertices in $W$
to the destinations. Concatenating these v-paths to admissible ``double-via''
paths would be comparable to the admissibility checks described in
this paper.

REVC seeks to identify all admissible paths between the given origins
and destinations. However, even if we do not apply approximations
(i.e. choose $\gamma=\delta=1$), some admissible paths may be falsely
rejected. This limitation is due to the preprocessing step, in which
shortcut edges are added to the graph, and the requirement that an
edge adjacent to the via vertex must be scanned in forward and backward
direction. However, we have already noted that the effect of the shortcut
edges can be arbitrarily reduced by imposing length constraints on
shortcut edges. Furthermore, most admissible paths will satisfy the
mentioned edge requirement (see Appendix \ref{sec:APX-Admissible-paths-excluded-by-using-double-scanned-edges}).
Therefore, these limitations generally have minor effects on the results.

REVC, as introduced in this paper, identifies identical paths based
on their lengths. Alternative approaches exist but might be less efficient.
In practice, distinct paths may have identical lengths, and REVC may
therefore falsely reject some admissible paths. Paths with equal lengths
occur most frequently in cities whose roads form a grid structure.
Nevertheless, since the roads may have distinct speed limits and traffic
volumes, and because turns take additional time, paths with identical
lengths may not occur frequently in practice. Since ties are even
less likely in long paths, we argue that it is reasonable to distinguish
paths based on their lengths.

Misclassifications of distinct paths with equal lengths can be reduced
by adding small random perturbations to the lengths of all edges.
Though this procedure makes it unlikely that admissible paths with
similar lengths are considered identical, the perturbation term randomly
defines an optimal path in grid networks. Therefore, the random perturbation
is of limited help in these networks. Note, however, that regardless
of how we identify identical paths, REVC and similar shortest-path-based 
methods are not well suited to work in grid networks, as ties
must be broken when the shortest path trees are grown.

An important feature of REVC is to reject u-turn paths by considering
via edges instead of via vertices. In undirected graphs, this procedure
also ensures that the returned routes do not contain cycles. However,
in directed graphs it is possible that locally optimal paths include
a cycle, and REVC may return such paths. Though it is unlikely that
locally optimal routes with cycles occur in realistic road networks,
it is possible to check paths for cycles before returning them. Confirming
that no vertices appear twice in a path can be done efficiently. 

From a modelling perspective, it may be desirable to restrict admissible
paths not only by excluding paths with cycles but to impose a more
general requirement instead. \citet{abraham_alternative_2013} suggest
to apply the relative length bound not only to entire paths but also
to their subsections (see Appendix \ref{sec:APX-Comparison-of-REV-REVC}).
This would exclude paths with subsections for which much better shortcuts
exist. However, testing this constraint involves relatively high
computational complexity, and there may also be situations in which
the additional requirement may not be of benefit in models. In any
event, the route sets returned by REVC can serve as a starting point
before further restrictions are applied.

In this paper, we presented performance measurements to assess the
efficiency of REVC and applied optimization procedures. When evaluating
these results, it is important to note the limitations of our implementation.
For example, our parallel implementation comes with scheduling overheads.
Some parts of the algorithm were not parallelized at all, leaving
room for further speedups. Furthermore, the slowdown factors we measured
can be considered as upper bounds, since we compared a highly optimized
shortest path search with a high-level implementation of REVC. Despite
these limitations, the most important timing result remains visible:
the performance of REVC scales well with the numbers of routes and
end points.

We conducted empirical tests showing that locally optimal paths can
be used to discern which roads are used by travellers of interest.
Though this is a strong indicator that local optimality criteria can
be successfully applied in route choice models, our test does not
provide final proof. On the one hand, we have surveyed traveller behaviour
at a small set of locations only and are thus unable to know how these
travellers behaved elsewhere. On the other hand, even if we knew that
travellers use only roads that are part of locally optimal routes,
we would not know how the travellers combine these roads. In addition
to the conceptual arguments and empirical results presented in this
paper, a more thorough analysis of empirical tracking data (see e.g.
\citealp{bekhor_evaluation_2006}) would be worthwhile. This will
remain a task for future research.

\section{Conclusion}

Generating route choice sets with locally optimal single-via paths
has a sound mechanistic justification, leads to small choice sets
with reasonable alternatives, and requires minimal data. We presented
an algorithm that efficiently generates such choice sets for large
numbers of origin-destination pairs. The algorithm is able to identify
(almost) all locally optimal single-via paths up to a specified length
between the origins and destinations. Therefore, the algorithm extends
earlier methods based on local optimality and makes the approach a
valuable method to generate route choice sets.

We confirmed that predictions made based on the algorithm's results
matched empirical traffic observations. Furthermore, we assessed the
algorithm's performance dependent on the input parameters. The results
provide insights into the effect of approximation parameters and the
distribution of locally optimal paths in real road networks. Therefore,
our study provides the necessary prerequisites to construct route
choice sets based on local optimality in large-scale traffic simulation
applications.

\section*{Acknowledgements}

The author would like to give thanks to Mark A. Lewis and his research
group at the University of Alberta for helpful feedback and discussions.
Furthermore, the author would like to thank Martina Beck and the staff
of the BC Invasive Mussel Defence Program for collecting and providing
the empirical data used in this study.

\section*{Additional information}
\begin{description}
\item [{Funding:}] This work was supported by the Canadian Aquatic Invasive Species Network 
and the Natural Sciences and Engineering Research Council of Canada.
\item [{Competing~interest:}] The author declares no competing interest.
\end{description}

\bibliographystyle{mee}

\newpage{}

\appendix

\part*{Appendix}

\setcounter{lem}{0}
\setcounter{cor}{0}
\setcounter{figure}{0}
\renewcommand{\theequation}{A\arabic{equation}}
\renewcommand{\thefigure}{A\arabic{figure}}

\section{Proofs\label{sec:APX-Proofs}}

In this Appendix, we prove Lemma \ref{lem:treeBoundLemma} and Corollary
\ref{cor:treeBound} (main text). We adjust the statement of Lemma
\ref{lem:treeBoundLemma} to recall notation from the main text.
\begin{lem}
Consider an arbitrary admissible single-via path $P$ from $s$ to
$t$. With $x_{s}'=\argmino{x\in P;\,{{\dd[P]sx}}\geq\alpha\l P}\dd[P]sx$,
let
\begin{eqnarray}
x_{s} & := & \begin{cases}
x_{s}' & \text{if }\dd[P]s{x_{s}'}\leq\frac{1}{2}\l P\\
\argmaxo{x\in P;\,{{\dd[P]sx}}\leq\frac{1}{2}\l P}\dd[P]sx & \text{else.}
\end{cases}\label{eq:x_t-def}
\end{eqnarray}
Choose $x_{t}$ accordingly. Then there is at least one vertex $v\in P$
with \label{lem:treeBoundLemma-1} 
\begin{enumerate}
\item $\dd[P]sv=\dd sv\leq\dd[P]s{x_{t}}$ \label{enu:treeBoundLemma1-1}
and
\item $\dd[P]vt=\dd vt\leq\dd[P]{x_{s}}t$.\label{enu:treeBoundLemma2-1}
\end{enumerate}
\end{lem}
\begin{proof}
Since $P$ is a single-via path, $P$ contains at least one vertex
$v'$ such that $\dd[P]s{v'}=\dd s{v'}$ and $\dd[P]{v'}t=\dd{v'}t$.
That is, $v'$ splits $P$ into two shortest paths. Now choose a vertex
$v$ as follows:
\begin{eqnarray}
v & := & \begin{cases}
v' & \text{if }\dd[P]s{v'}\leq\dd[P]s{x_{t}}\text{ and }\dd[P]{v'}t\leq\dd[P]{x_{s}}t,\\
x_{t} & \text{if }\dd[P]s{v'}>\dd[P]s{x_{t}},\\
x_{s} & \text{if }\dd[P]{v'}t>\dd[P]{x_{s}}t.
\end{cases}\label{eq:treeBoundLemmaV}
\end{eqnarray}

We show that $v$ satisfies the lemma's requirements by regarding
the different possible choices of $v$:
\begin{enumerate}
\item If $\dd[P]s{v'}\leq\dd[P]s{x_{t}}$ and $\dd[P]{v'}t\leq\dd[P]{x_{s}}t$,
then the conditions \ref{enu:treeBoundLemma1-1} and \ref{enu:treeBoundLemma2-1}
are clearly satisfied for $v:=v'$.
\item If $\dd[P]s{v'}>\dd[P]s{x_{t}}$, then inserting $v:=x_{t}$ yields
$\dd[P]s{v'}>\dd[P]sv$. Therefore, the subpath $P^{sv}$ from $s$
to $v$ is a subpath of the subpath $P^{sv'}$ from $s$ to $v'$.
Since $v'$ splits $P$ into two shortest paths, $P^{sv'}$ is a shortest
path. Therefore, $P^{sv}$ must be a shortest path, too. Thus, $\dd[P]sv=\dd sv=\dd[P]s{x_{t}}$,
and condition \ref{enu:treeBoundLemma1-1} is satisfied. \\
To show that condition \ref{enu:treeBoundLemma2-1} holds as well,
observe that $\dd[P]vt=\dd[P]{x_{t}}t\leq\frac{1}{2}\l P\leq\l P-\dd[P]s{x_{s}}=\dd[P]{x_{s}}t$.
It remains to be shown that $\dd[P]vt=\dd vt$. Since $P$ is $\alpha$-relative
locally optimal, each subpath whose length after removal of one end
point would be smaller than $\alpha\l P$ is a shortest path. By construction,
this applies to the subpath from $x_{t}$ to $t$. Hence, it is $\dd[P]vt=\dd vt$
and condition \ref{enu:treeBoundLemma2-1} is satisfied.\label{enu:treeBoundLemmaArg-1}
\item The proof for the case $\dd[P]{v'}t>\dd[P]{x_{s}}t$ is analogous
to the argument presented under point \ref{enu:treeBoundLemma2-1}.
\end{enumerate}
\end{proof}
\begin{cor}
\label{cor:treeBound-1}For each admissible v-path between an origin-destination
pair $\left(s,t\right)$, a via vertex will be scanned from both directions
if the shortest path trees are grown up to a height of 
\begin{eqnarray}
h_{\text{max}} & := & \max\left\{ \left(1-\alpha\right)\beta\l{P_{st}},\,\frac{1}{2}\beta\l{P_{st}}\right\} .
\end{eqnarray}
\end{cor}
\begin{proof}
Let $P$ be an admissible path, which implies that $\l P\leq\beta\l{P_{st}}$.
Recall that 
\begin{eqnarray}
x_{t}' & = & \argmino{x\in P;\,{{\dd[P]xt}}\geq\alpha\l P}\dd[P]xt\nonumber \\
 & = & \argmino{x\in P;\,\l P-{{\dd[P]sx}}\geq\alpha\l P}\left(\l P-\dd[P]sx\right)\nonumber \\
 & = & \argmaxo{x\in P;\,{{\dd[P]sx}}\leq\left(1-\alpha\right)\l P}\dd[P]sx.
\end{eqnarray}
Therefore, $x_{t}$ is either the last vertex in $P$ with $\dd[P]sx\leq\left(1-\alpha\right)\l P\leq\left(1-\alpha\right)\beta\l{P_{st}}$
or the last vertex with $\dd[P]sx\leq\frac{1}{2}\l P\leq\frac{1}{2}\beta\l{P_{st}}$
(see equation (\ref{eq:x_t-def})). Either way, $x_{t}$ will be included
in the shortest path tree if we grow the tree to a height of just
above $\max\left\{ \left(1-\alpha\right)\beta\l{P_{st}},\,\frac{1}{2}\beta\l{P_{st}}\right\} $.
The same argument holds in backward direction for $x_{s}$. From Lemma
\ref{lem:treeBoundLemma-1} we know that $P$ is a v-path via a vertex
$v\in P^{x_{s}x_{t}}$ located between $x_{s}$ and $x_{t}$. Since
both $x_{s}$ and $x_{t}$ are scanned from both sides, the vertex
$v$ will be scanned from both sides as well.
\end{proof}

\section{Admissible paths excluded by requiring that a neighbouring edge of
the via vertex has been scanned from both directions\label{sec:APX-Admissible-paths-excluded-by-using-double-scanned-edges}}

Requiring that a neighbouring edge of the via vertex has been scanned
in both directions excludes u-turns without reducing the number of
found admissible paths significantly. However, there is exactly one
scenario in which an admissible v-path is not found if we impose this
constraint. The situation is depicted in figure \ref{fig:via-vertices-via-edges-bad}.

Suppose the v-path $P$ from $s$ to $t$ via the vertex $v$ is admissible
but falsely rejected by the exact version of REVC ($\gamma=\delta=1$).
Suppose furthermore that $u\in P$ is the predecessor of $v$ and
$w\in P$ the successor. Then there must be a vertex $x\in P^{su}$
and a vertex $y\in P^{wt}$ such that the following conditions hold:
\begin{enumerate}
\item The shortest path from $x$ to $w$ does not include $v$: $\dd xv+\dd vw>\dd xw$.
\item The shortest path from $u$ to $y$ does not include $v$: $\dd uv+\dd vy>\dd uy$.
\item Let $x'$ be the direct successor of $x$ in $P$. It must be $\dd{x'}v>\alpha\cdot\l P$.
\item Let $y'$ be the direct predecessor of $y$ in $P$. It must be $\dd v{y'}>\alpha\cdot\l P$.
\item The shortest path from $u$ to $w$ must include $v$: $\dd uw=\dd uv+\dd vw$.
\end{enumerate}
If the first two conditions were not satisfied, at least one edge
on $P$ adjacent to $v$ would be scanned from both directions and
$P$ would be found. If the last three conditions were not satisfied,
$P$ would not be admissible.

Though it is possible that all of these conditions are satisfied,
we believe that such a scenario is unlikely in real road networks. 

\begin{figure}
\begin{centering}
\includegraphics[scale=0.6]{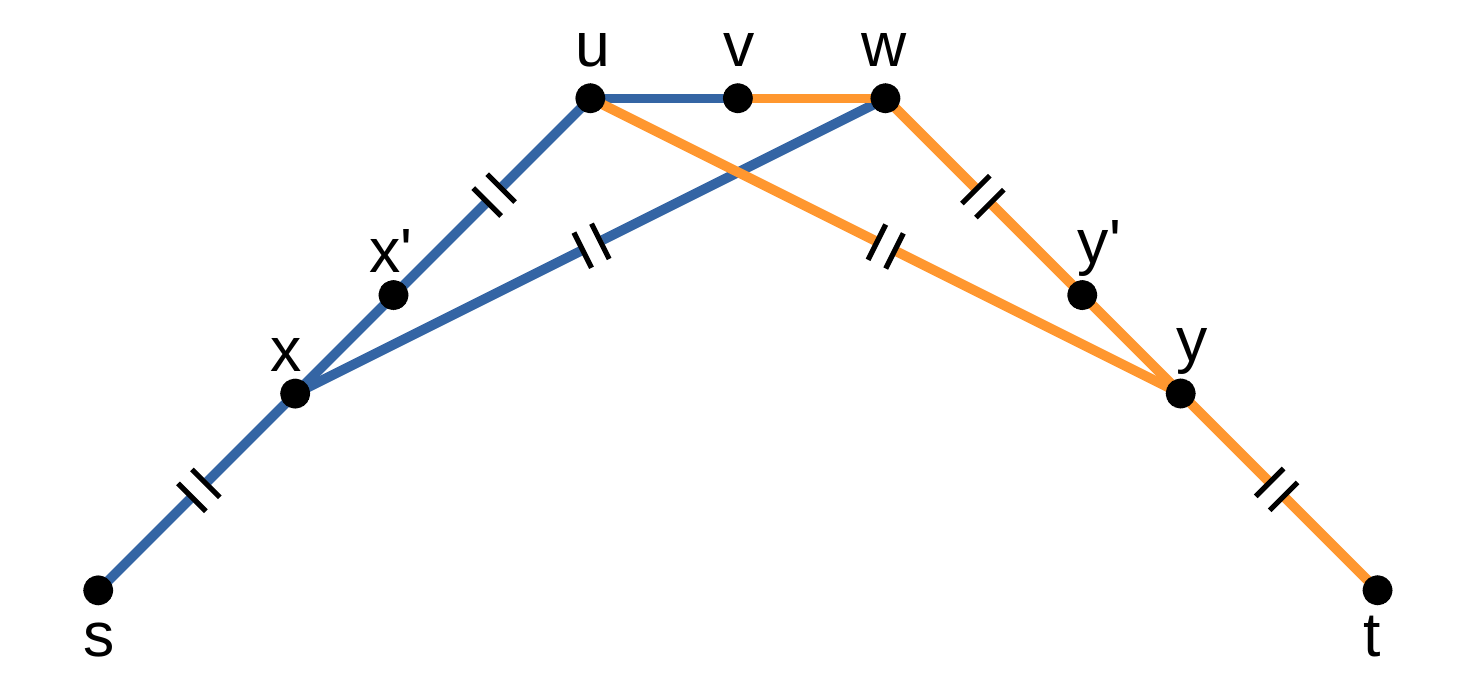}
\par\end{centering}
\caption[Scenario in which an admissible path is excluded due to the requirement
that an edge adjacent to the via vertex is scanned in both directions]{Scenario in which an admissible path is excluded due to the requirement
that an edge adjacent to the via vertex is scanned in both directions.
Blue lines depict the edges included in the forward shortest path
tree grown from the origin $s$ and orange lines the edges of the
backward tree grown into the destination $t$. Lines that may represent
multiple edges are indicated with a gap. As the edges adjacent to
$v$ are included in one shortest path tree only, the path $P_{svt}$
would be rejected by REVC. \label{fig:via-vertices-via-edges-bad}}
\end{figure}

\begin{rem}
It can be shown that pruning does not weaken these conditions.
\end{rem}

\section{Comparison of REV and REVC\label{sec:APX-Comparison-of-REV-REVC}}

In this Appendix, we compare our algorithm REVC to the algorithm REV
\citep{abraham_alternative_2013} that it is based on. To a large
extent, REVC uses the same ideas as REV: shortest path trees are grown
around the origin and destination, and v-paths via vertices scanned
from both directions are checked for admissibility using an approximate
test for local optimality. However, REV and REVC differ in (1) the
admissibility definition (2) the choice of the returned paths, and
(3) technical optimizations that REVC introduces. Below we discuss
each of these points.

\subsection{Admissibility definition\label{subsec:Admissibility-definition}}

The admissibility definition by \citet{abraham_alternative_2013}
includes three requirements. They say a v-path $P_{svt}$ is admissible
if
\begin{enumerate}
\item $P_{svt}$ has limited overlap with previously identified admissible
paths $P_{swt}$ between $s$ and $t$. That is, $\l{P_{svt}\cap\left(\underset{w}{\cup}P_{swt}\right)}\leq\eta\cdot\l{P_{st}}$.
\item $P_{svt}$ is $T$-locally optimal with $T=\alpha\cdot\l{P_{st}}$.
\item $P_{svt}$ has $\beta$-uniformly bounded stretch. That is, for all
$u,w\in P_{svt}$, it is $\l{P_{svt}^{uw}}\leq\beta\cdot\l{P_{uw}}$.
\end{enumerate}
None of these requirements coincides exactly with the constraints
we imposed in our paper.

Requirement 1 does not appear in our admissibility definition. The
constraint requires that the admissible paths have a clearly specified
order. However, though \citet{abraham_alternative_2013} suggest a
reasonable ordering, this introduces another degree of freedom whose
impact on the results may be oblique. Furthermore, we were interested
in identifying \emph{all} routes that satisfy certain criteria and
leave it to the second modelling stage, in which a route is chosen
from the choice set, to take route overlaps into account \citep[see e.g.][]{cascetta_modified_1996}.
Lastly, the local optimality criterion naturally limits the pair-wise
overlap of paths. Therefore, we dropped this constraint.

Requirement 2 differs from our local optimality constraint, because
the length $T$ of the subsections required to be optimal depends
on the shortest distance between $s$ and $t$ rather than the length
of the via path. This allows for more admissible paths. We changed
this requirement for two reasons: (1) the spatial scale at which travellers'
decision routines change is likely dependent on the path they \emph{actually}
choose rather than the shortest alternative, which may -- dependent
on the global quality metric -- not even be a favourable option.
Travellers on a long trip may have a higher incentive to choose a
route with long optimal subsections. (2) The adjusted local optimality
criterion allows for more effective pruning with simpler bounds when
considering many origin-destination pairs. Using a pair-wise static
local optimality criterion as \citet{abraham_alternative_2013} would
require us to choose the pruning bound dependent on the origin-destination
pair closest together. For these reasons, we introduced the notion
of relative local optimality. Note that REVC can also be used to identify
all paths satisfying requirement 2 if the constant $\alpha$ is adjusted
accordingly and the resulting paths are filtered so that suboptimal
paths are excluded.

Requirement 3 is relaxed in our admissibility definition. \citet{abraham_alternative_2013}
do not introduce an efficient algorithm to identify paths satisfying
requirement 3. Instead of bounding the lengths of all subpaths, they
consider the complete path only, as we do in this paper. Nonetheless,
uniformly bounded stretch is a valuable characteristic for choice
set elements. However, since REVC will return a moderate number of
paths in many applications, paths could be checked for uniformly bounded
stretch after execution of REVC. Consequently, we have used the relaxed
constraint directly.

\subsection{Returned paths\label{subsec:Returned-paths}}

\citet{abraham_alternative_2013} aim to compute a small number of
high-quality paths between an origin and a destination efficiently.
To save computation time, they do not assess the admissibility of
all path candidates. Instead, REV processes the potentially admissible
paths in an order dependent on some objective function, estimating
the quality of the paths. REV returns the first $n$ processed approximately
admissible paths.

Since we are interested in an exhaustive search for admissible paths,
we do not process the paths in a specific order. We return all approximately
admissible paths and leave the assessment of their quality, if desired,
to a second, independent algorithm.

\subsection{Optimizations}

REVC introduces multiple optimization to REV. First, REVC uses a tighter
bound for the tree growth and the pruning stage. Though our pruning
bound would have to be adjusted to comply with the admissibility definition
applied by \citet{abraham_alternative_2013} (see section \ref{subsec:Admissibility-definition}),
the ideas introduced in this paper are still applicable.

Second, REVC excludes u-turns by considering via edges rather than
via vertices. Furthermore, REVC identifies vertices representing identical
paths before assessing their admissibility. Both optimizations could
be directly applied to speed up REV. However, REV processes the paths
in an order given by some objective function (see section \ref{subsec:Returned-paths}).
It is possible to construct this objective function so that u-turn
paths are not processed before the admissible paths.

Third, to control the accuracy of the results, REVC uses the $\T[\delta]$-test
instead of the $\T$-test to check whether a path is locally optimal.
This optimization could also be applied in REV, though it may effect
the performance of REV more strongly than the performance of REVC.

Lastly, REVC is optimized to process many origin-destination pairs
at once. Though the idea to grow each shortest path three only once
per origin and destination is straightforward, the main innovation
of REVC is in the efficient local optimality checks of many v-paths
via one via vertex.

\end{document}